\documentclass[letterpaper, 11pt]{article}
\usepackage{amssymb}
\usepackage{amsmath}
\usepackage{natbib}
\usepackage{color}
\usepackage{listings}
\usepackage{graphicx,enumerate}
\usepackage[space]{grffile}
\usepackage{subcaption}
\usepackage{mathtools}
\usepackage{caption}
\usepackage{graphicx,xcolor,colortbl}    
\usepackage{epstopdf,epsfig} 
\usepackage{hyperref}
\usepackage{IEEEtrantools}
\hypersetup{ 
  colorlinks=false,
  pdfborder = {0 0 0.5 [3 0]}
}

\newtheorem{theorem}{Theorem}

\newtheorem{example}[theorem]{Example}

\newtheorem{lemma}[theorem]{Lemma}

\newtheorem{proposition}[theorem]{Proposition}

\numberwithin{equation}{section} 
\newenvironment{proof}[1][Proof]{\textbf{#1.} }{\ \rule{0.5em}{0.5em}}
\def\R{{\mathbb R}}        
\def\P{{\mathbb P}}        
\def\E{{\mathbb E}}        
 
\def\1{{\mathbf 1}}        
\def\F{{\mathcal F}}        

\def\var{{\mathop{\mathbf{Var}}}}    
\def\L{{\mathcal L} \,}

\def\setT{{\mathcal T}}

\def\T{{\bar{T}}}

\def\Vk1{{V^{(i-1)}}}

\def\Hk1{{H^{(i-1)}}}

\def\pk1{{p^{(i-1)}}}

\usepackage{multirow}
\usepackage[left=1in,right=1in,top=1in,bottom=1.1in]{geometry}

\begin{document}
\title{Optimal Timing to Trade Along a Randomized Brownian Bridge}
\author{Tim Leung\thanks{Department of Applied Mathematics, University of Washington, Seattle WA 98195. E-mail:
\mbox{timleung@uw.edu}.  Corresponding author.} \and Jiao Li\thanks{APAM Department, Columbia University, New York, NY 10027; email:\,\mbox{jl4170@columbia.edu}.} \and Xin Li\thanks{Bank of America Merrill Lynch, One Bryant Park, New York, NY 10036; email:\,\mbox{xinli.columbia@gmail.com}. } }
\date{\today}
\maketitle
\begin{abstract}
This paper studies an optimal trading problem that incorporates the trader's market view on the terminal asset price distribution and uninformative noise embedded in the asset price dynamics. We model the underlying asset price evolution   by an exponential randomized Brownian bridge (rBb) and consider various prior distributions for the random endpoint.  We solve for the optimal strategies to sell a stock, call, or put, and analyze the associated delayed liquidation premia.  We solve for the optimal trading strategies numerically and compare them across different prior beliefs. Among our results, we find that disconnected continuation/exercise regions arise when the trader prescribe a two-point discrete distribution and double exponential distribution. 
\end{abstract}
\vspace{10pt}

 \newpage
\section{Introduction}
One fundamental problem faced by  all traders is to determine when to sell an asset or financial derivative over a given trading horizon. The optimal trading decision  depends crucially on the trader's subjective belief of the distribution of the asset price in the future and the observed price fluctuations. By monitoring the asset price evolution over time, the trader decides whether to sell the asset or derivative at the prevailing market price, or continue to wait till a later time.  In this paper, we tackle this problem by  constructing models that reflect    the two major features: (i) the trader's market view on the terminal asset price distribution and the uninformative noise embedded in the asset price dynamics, and (ii) the timing option that gives rise to an optimal stopping problem corresponding to the trader's asset or derivative liquidation. 

In order to  describe  the    asset price evolution,  we present a  randomized Brownian bridge (rBb)  model, whereby the log-price of the asset follows a Brownian bridge with a randomized endpoint representing the random  terminal log-price. In turn, the trader's prior belief and learning mechanism can be   encoded in the drift of the log-price process.  Our model is adapted from the novel work by \cite{cartea2013algorithmic}. They model the asset's mid-price as a randomized Brownian bridge, and solve a   trading problem that maximizes the expected trading revenue by placing either limit or market orders while penalizing the running inventory. In comparison, the underlying asset price in our model is an exponential rBb with different randomized endpoints, with examples including discrete, normal, and double exponential    distribution. Also,  we introduce an optimal stopping approach to the trader's liquidation problem, which is applicable not only to selling the underlying asset, but also options written on it.

To solve the trader's optimal stopping problem, we devise and apply a number of analytical tools. We first define the  optimal liquidation premium that  represents the additional value from optimally waiting to sell, as opposed to immediate liquidation. Indeed, as soon as this premium vanishes it is optimal for the trader to sell. On the other hand, if this premium is always strictly positive, then the trader finds it optimal to wait through maturity. Therefore, the optimal liquidation premium provides not only new financial interpretations, but also another avenue of analytical investigation of the optimal stopping problem. Among our results,  we  identify the conditions under which it is optimal to immediately liquidate, or hold the asset/option position through expiration. Furthermore, we prove that the optimal strategy to liquidate a long-call-short-put position is identical to the optimal strategy to sell the underlying stock under the same trading horizon. This timing parity holds for any distribution of the randomized endpoint.   Moreover, we  derive the  variational inequality associated with the optimal stopping problem, and present a finite-difference method to solve for the optimal trading boundaries.

In the literature,  Brownian bridges  have been  used to represent market uncertainty or uninformative noise (see e.g. \cite{brody2008informed}, \cite{hughston2012pricing}, and \cite{macrina2014heat}).  The randomized Brownian bridge model in this paper belongs to the information-based approach to pricing and trading. 
Among related studies,  \cite{brody2008information} and \cite{filipovic2012conditional} study information-based models where the asset prices are computed via conditional expectation with respect to an information process. 

Randomized Brownian bridges, or their variations, have great potential applicability in a number of  finanical applications. For example, while futures prices are supposed to be equal to the spot price in theory, it is observed that  some commodity futures prices do not exactly converge to  the corresponding spot prices upon maturity (see \cite{guoleung17}). \cite{futuresBS} and \cite{futuresDaiKwok} investigate arbitrage strategies on stock index futures where the index arbitrage basis is assumed to follow a Brownian Bridge.  One can also look for more potential applications in index tracking and  exchange-traded funds  (see \cite{leung2016leveraged}). The randomized endpoint provides added flexibility in modeling random shocks to the asset price on a future date, which is   applicable to events such as  Federal Reserve announcements, and  earnings surprises (see e.g. \cite{johannes2006earnings} and \cite{leung2014accounting}). 

As for optimal stopping problems involving  Brownian bridges, \cite{ekstrom2009optimal} consider   optimal single stopping  of a Brownian bridge or odd powers of a Brownian bridge, without discounting. \cite{ekstrom2017optimal} use an  optimal stopping approach to maximize the expected value of  a Brownian bridge with an unknown pinning point. The optimal double problem is further studied in \cite{baurdoux2015optimal}, where they maximize the expected spread between  the payoffs at the entry and exit times where the underlying follows a Brownian bridge.

In related works on optimal stopping problems for securities trading under a finite horizon setting, \cite{LeungLudkovski2011,LeungLudkovski2} introduce the concept of delayed purchase premium and   analyze the problem of purchasing equity European and American options in an incomplete market, where the investor's belief on asset price dynamics may differ from prevalent market views. \cite{LeungLiu2012} analyze the delayed purchase premium associated with credit derivatives under a multi-factor intensity-based default risk framework, and derive the optimal trading strategies.  In \cite{LeungShirai}, the authors study the optimal timing to sell an asset or option subject to a path-dependent risk penalty.

The rest of the paper is structured as follows.   In Section \ref{sect-stockdynamics}, we present  the randomized Brownian bridge model for the underlying asset price, show how different prior beliefs are encapsulated in the stochastic differential equation corresponding to the price dynamics, and illustrate the price behaviors through Monte Carlo simulation. In Section \ref{sect-optimalproblem}, we formulate and analyze the optimal liquidation problem and present numerical results to illustrate the optimal trading strategies. In Section \ref{sect-finitedifference}, we summarize  the numerical algorithm used for  solving the optimal stopping problems under different settings in this paper. A number of  proofs are collected in Section \ref{proofs}.

\cite{oshima2006optimal}, \cite{peskir2006optimal}

\section{Prior Belief and Price Dynamics}\label{sect-stockdynamics}
The model consists of a single asset whose positive price process is denoted by $(S_t)_{0\le t\le T}$. The  trader in our model specifies a  prior belief on the future distribution of the stock price at a fixed future time $T$. We denote  $X_t =\log S_t$ to be the log-price of the asset, with $X_0$ being the stock's initial log-price. The trader's belief is described by a real-valued random variable $D$ to be realized at time $T$ so   that the terminal log-price   is given by
\begin{align}
X_T = X_0 + D.
\end{align} 
To avoid arbitrage, we require that $D$ have finite second moment and $\P(D > rT) \in (0, 1)$, where $r$ is positive risk-free interest rate, under the historical probability measure $\P$. 

As time progresses, new information arrives in the form of changes in asset price. The price fluctuation can also be viewed as noise prior to the realization of the terminal log-price $D$.   We model the log-price process as a randomized Brownian bridge. To this end, we first let  $(\beta_{t})_{0\le t\le T}$ be a standard Brownian bridge   \begin{align}
\beta_{t} = B_t - \frac{t}{T}B_T, \quad t \in [0,T],
\end{align} 
where  $(B_t)_{0\le t \le T}$ is a standard Brownian motion.  The process $\beta$ starts at $0$ and ends at $0$. It can be viewed as uninformative noise, which vanishes both at times $0$ and $T$. We assume that the processes  $B$ and thus  $\beta$ are independent of the   random variable $D$. 

Then, the  log-price process  is given by 
\begin{align}
X_t  = X_0 + \sigma \beta_{t} + \frac{t}{T} D, \quad t \in [0,T], \label{Xpath}
\end{align} 
where $\sigma > 0$ is a constant parameter.  The trader observes the price evolution of the asset, and the corresponding filtration  $\mathbb{F}\equiv (\F _t)_{0\le t\le T}$ is  generated by the  log-price process  $X$. In other words, the trader does not directly observe the standard Brownian bridge $\beta$ over time. Mathematically,  this means that $\beta _t$ is not  adapted to $\mathbb{F}$. The financial interpretation is that  $\beta$ reflects the  uninformative noise in the markets such as market sentiments and rumors. Without direct observation of  $\beta_t$, the trader is unable to  decompose $X_t$ into   $\beta_t$ and $D$ at any time $t<T$. At time $T$, the trader observes the  realization of $D$ and thus $X_T$. The parameter $\sigma$ allows us to control effect of $\beta$ on the log-price  fluctuation. By inspecting  equation \eqref{Xpath}, the ratio $\frac{t}{T}$ can be viewed as the rate at which the value of $D$ is revealed, going linearly from fully hidden at time $0$ to fully revealed at time $T$. 

By Proposition 1 of \cite{cartea2013algorithmic},   the log-price process satisfies the stochastic differential equation (SDE)
\begin{align}
dX_t = A(t,X_t)dt + \sigma dW_t,\label{LogPrice}
\end{align}
for $t\in[0,T]$, where $(W_t)_{0\le t\le T}$ is an $\mathbb{F}$-adapted standard Brownian motion  under the probability measure $\P$.
The drift term is
\begin{align}
A(t,X_t) = \frac{a(t,X_t) - (X_t - X_0)}{T - t},\label{AA}
\end{align}
where 
\begin{align}
a(t,x) := \E[D | X_t = x] = \frac{\int_{-\infty}^\infty z\text{exp}\left(z\frac{x - X_0}{\sigma^2(T-t)} - z^2\frac{t}{2T\sigma^2(T-t)}\right)dF(z)}{\int_{-\infty}^\infty \text{exp}\left(z\frac{x-X_0}{\sigma^2(T-t)} - z^2\frac{t}{2T\sigma^2(T-t)}\right)dF(z)},\label{a}
\end{align}
and $F(\cdot)$ is the cumulative distribution function of the random variable $D$. We refer to Appendix A of \cite{cartea2013algorithmic} for the derivation of \eqref{LogPrice}. The $\mathbb{F}$-Brownian motion $W_t$ appearing in \eqref{LogPrice} can be considered as the market information accessible by the trader. Unlike $\beta_t$, the value of $W_t$ contains real information relevant to the risky asset price. \eqref{LogPrice}-\eqref{a} represent how the price innovations reflect the   probability distribution of $D$ and market information flow. The incomplete information in \eqref{Xpath} is reformulated as a complete information by projecting the log-price innovations onto the observable filtration. In that sense, the information-based approach is more flexible to add additional interpretation and intuition based on observable price process, as we present in the following sessions of the paper.

To understand the  conditional expectation in  \eqref{a}, we start with the definition
\begin{align}
\E [D | X_t = x] = \int _{-\infty} ^{\infty} z \pi _t(z) dz,\label{edxx}
\end{align}
where $\pi _t(z)$ is the conditional probability density or mass function for the random variable $D$ defined by 
\begin{align}
\pi _t(z) = \frac{d}{d z} \P(D \leq z | X_t = x).\nonumber
\end{align}
Using Bayes formula, the conditional probability density is given by 
\begin{align}
\pi _t(z) = \frac{p(z)\rho_{X_t}(x | D = z)}{\int _{-\infty} ^{\infty} p(z)\rho_{X_t}(x | D = z)dz },\label{bayes}
\end{align}
where $p(z)$ denotes the probability density or mass function for  $D$, and $\rho_{X_t}(x | D = z)$ denotes the conditional density function for the random variable $X_t$ given  $D = z$. According to \eqref{Xpath}, for any fixed $t\in [0,T)$, $X_t$ given $D=z$ is Gaussian, i.e. \[X_t\big{|}_{D=z} \sim \mathcal{N} \left( X_0 + \frac{t}{T} z  \,, \,  \sigma ^2 \frac{t}{T} (T - t)  \right)\,.\]  Therefore, we can express  the conditional probability density for $X_t$  as
\begin{align}
\rho_{X_t}(x | D = z) = \frac{1}{\sigma \sqrt{2 \pi \frac{t(T - t)}{T}}} \exp\left( - \frac{T (x - X_0 - \frac{t}{T} z) ^2}{2 \sigma ^2 t (T - t)}\right). \nonumber
\end{align}
Substituting this expression into \eqref{bayes}, we have
\begin{align}
\pi _t (z) = \frac{p(z) \exp \big(z \frac{x - X_0}{\sigma ^2 (T - t)} -  z^2 \frac{t}{2T \sigma ^2 (T - t)}\big)}{\int _{-\infty} ^{\infty} p(z) \exp \big(z \frac{x - X_0}{\sigma ^2 (T - t)} -  z^2 \frac{t}{2T \sigma ^2 (T - t)}\big) dz }.  \label{pitz}
\end{align}
Finally, equation \eqref{a} follows after we substitute \eqref{pitz} into \eqref{edxx}.

The  asset price is an exponential rBb  defined by  $S_t = \exp(X_t)$. By  Ito's lemma, we obtain the SDE 
\begin{align}
d S_t = \left(A(t, X_t) + \frac{\sigma^2}{2}\right)S_t \,dt +  \sigma S_t  \,dW_t, \quad  0\le t\le T.\end{align}
  In this paper,  we will consider  three different distributions for $D$: two-point discrete distribution, normal distribution, and double exponential distribution. The first one is discrete while the other two  are continuous distributions. Compared to the   normal distribution, the  double exponential distribution is capable of generating two heavier tails that can be symmetric or asymmetric. Next, let us illustrate the effect of the distribution of $D$ on the  asset price  dynamics.

\begin{example}\label{example1}
Suppose that  the trader's prior belief on the future log-price  follows a two-point discrete distribution  defined by 
\begin{align}
D_{\delta} = \begin{cases} \delta_u \text{ with probability } p_u, \\ \delta_d \text{ with probability } p_d = 1-p_u,\end{cases}
\end{align}
with mean and variance, respectively,
\begin{align}
\E[D_{\delta}] &= \delta_u p_u + \delta_d p _d, \\
\var[D_{\delta}] &= \delta_u ^2 p_u (1 - p_u) + \delta_d ^2 p_d (1 - p_d) - 2 \delta_u \delta_d p_u p_d .
\end{align}
Then,  the drift term $A(t, x)$ of $X_t$ in \eqref{LogPrice} is given by
\begin{align}
A(t, x) = \frac{1}{T - t}\left[\frac{\delta_u u(t, x) + \delta_d d(t, x)}{u(t, x) + d(t, x)} - (x - X_0)\right]\,, 
\label{A_twopoint}
\end{align}
where
\begin{align}
u(t, x) = p_u \exp \left( \delta_u \frac{x - X_0}{\sigma^2(T-t)} - \delta_u^2\frac{t}{2T\sigma^2(T - t)} \right),
\end{align}
\begin{align}
d(t, x) = p_d \exp \left( \delta_d \frac{x - X_0}{\sigma^2(T-t)} - \delta_d^2\frac{t}{2T\sigma^2(T-t)}\right).
\end{align}
 
\end{example}

A  special case of Example \ref{example1} is when   $D$ is a constant $\delta$. Then,  the log-price process \eqref{LogPrice} is simply a Brownian bridge that starts  at $X_0$ and ends at $X_0 + \delta$. Its SDE is given by \begin{align}
dX_t = \frac{X_0 + \delta - X_t}{T-t} dt + \sigma dW_t, \quad 0\le t\le T.
\end{align}

\begin{example}\label{example2}
Suppose that  the trader's prior belief on the future log-price follows a normal distribution with mean $\mu$ and variance $\sigma_D^2$, i.e. 
\begin{align}
D_{n} \sim \mathcal{N}\big(\mu, \sigma_D^2\big).
\end{align}
Note that $\sigma _D$ is different from $\sigma$ in \eqref{LogPrice}. 
The drift of $X_t$ in \eqref{LogPrice} is given by
 \begin{align}
A(t, x) = \frac{(\sigma_D ^2 - T \sigma ^2) (x - X_0)  + \mu \sigma^2 T}{t \sigma_D ^2 + T \sigma^2 (T-t)}.
\label{A_normal}
\end{align}
The derivation is presented in Section \ref{CompNormal}.
\end{example}

As an alternative continuous distribution to the normal distribution, we now consider the double exponential distribution, which has been used to model random jumps in asset prices (see e.g. \cite{kou2002jump}).

\begin{example}\label{example3}  Suppose that  the trader's prior belief on the future log-price is represented by a double exponential random variable $D_e$ with  the  pdf
\begin{align}
f(z) = \mathbf{1}_{\{ z < \theta\}} p_1\lambda_1 e^{\lambda_1 (z -\theta)} + \mathbf{1}_{\{ z \geq \theta\}} p_2\lambda_2 e^{-\lambda_2 (z -\theta)},
\end{align}
where $p_1, p_2 >0$ and $p_2 = 1 - p_1$. The mean and variance of $D_e$ are, respectively, 
\begin{align}
\E[D_e] &= \theta - \frac{p_1}{\lambda _1} + \frac{p_2}{\lambda _2},\label{double_mean} \\ 
\var[D_e] &= \frac{2 p_1}{\lambda _1 ^2} + \frac{2 p_2}{\lambda _2 ^2} - (\frac{p_1}{\lambda _1} - \frac{p_2}{\lambda _2})^2.\label{double_var}
\end{align}
The drift of $X_t$ in \eqref{LogPrice} is given by
\begin{align}
A(t, x) = \frac{1}{T - t}\left[\frac{\sum _{i=1,2}  N_i (t, x) }{\sum _{i=1,2}H_i (t, x)} - (x - X_0)\right]\,,
\label{A_double}
\end{align}
where  
\begin{align}
\begin{cases} \displaystyle
\begin{split}
N_i (t, x) &= (-1)^{i}  p_i \lambda_i  \frac{1}{2\zeta} \exp \left( - \zeta(\theta + \frac{b_i}{2\zeta})^2 \right) \exp (\frac{b_i^2 - 4 \zeta  c_i}{4 \zeta }) \notag\\
&\quad -  p_i \lambda_i  \frac{b_i}{2 \zeta } \frac{\sqrt{\pi}}{\sqrt{ \zeta }} \exp (\frac{b_i^2 - 4 \zeta c_i}{4 \zeta })  \Phi((-1)^{i-1} d_i),  \\ 
\\ H_i (t, x) &= p_i \lambda_i \frac{\sqrt{\pi}}{\sqrt{ \zeta }} \exp (\frac{b_i^2 - 4 \zeta c_i}{4 \zeta }) \Phi((-1)^{i-1} d_i),
\end{split}	
\end{cases} 	
\end{align}
and
\begin{align*}
\begin{cases} 
\begin{split}
\zeta &\equiv \zeta (t) = \frac{t}{2T\sigma^2 (T-t)},  \\ 
\\ 
b_i &\equiv b_i (t, x) = -\left( \frac{x - X_0}{\sigma^2 (T-t)}  + (-1)^{i-1} \lambda_i \right),\\
\\ 
c_i &= (-1)^{i-1} \lambda_i \theta,\\
\\
d_i &\equiv d_i (t,x) = \sqrt{2 \zeta (t)} \left(\theta + \frac{b_i (t, x)}{2\zeta (t)} \right),
\end{split}	
\end{cases}		
\end{align*}
and 
$\Phi(\cdot)$ is the standard normal cumulative distribution function.  We present the computations in Section \ref{CompDouble}.
\end{example}

The paths of the log-price process $X$ can be simulated over discrete times using the standard Euler-Maruyama method.  Denote $\delta t$ as the discretization step, and set $t_i = i\delta t$, for $i=0, 1, 2, \ldots$. Then, the path of $X$, starting at $X_0$ at time $0$,  can be simulated iteratively as follows:  
\begin{align}
X_{t_{i+1}} =  X_{t_i} + A(t_i, X_{t_i}) \delta t + \sigma \sqrt{\delta t} \, \epsilon_i, \label{simulationX}
\end{align}
where    $(\epsilon_i)_{i = 1, 2,\ldots}$ is a sequence of IID  $\mathcal{N}(0,1)$ random variables. This method does not require the simulation of the random variable $D$ directly because the distributional  characteristics of $D$ are encapsulated in the drift function $A(t,x)$ (see \eqref{A_twopoint}, \eqref{A_normal} and \eqref{A_double}). The procedure takes in the current value of $X_{t_i}$ at each time step to compute the drift $A(t_i, X_{t_i})$, and each path evolves without the knowledge of  the Brownian bridge $\beta$ or the realization of the random variable $D$ as both $\beta$ and $D$ are never simulated. This is consistent with the fact that the trader cannot observe $D$ prior to $T$, and only has knowledge of the log-price process $X$ over time but not the Brownian bridge $\beta$. The paths will end up with  the terminal distribution resembling that of $D$. In Example \ref{example1},  under two-point discrete distribution for $D$, simulating according to \eqref{simulationX} will generate a path ending up at either $X_0 + \delta _u$ or $X_0 + \delta _d$, each with probability $p_u$ and $p_d$ respectively. For instance, in  Figure \ref{x_vs_t_sde}(a), we have $p_1 = p_2 = 0.5$, meaning that  about half of the paths will end up at either $X_0 + \delta _u$ or $X_0 + \delta _d$.

As seen in Examples \ref{example1}-\ref{example3}, the specification of the terminal log-price distribution directly affects the drift of $X_t$. To better observe the different structures of the drift   $A(t,x)$ under different distributions for $D$, we plot in  Figure \ref{A_compare} the function  $A(t,x)$ at $t= 0.1$ (panel (a)) and $t=0.8$ (panel (b)) with the three distributions sharing the same mean and variance.  Under the normal distribution, $A(t,x)$ is linear in $x$. In contrast,  it is neither linear nor monotone under the two-point discrete   and double exponential distributions. Moreover, under the two-point   distribution, $A$ is positive  when the stock price is low and negative when the stock price is relatively high, meaning that the asset price tends to have positive drift when price is low and negative drift when price is high.  However, the opposite is observed in $A$ under the normal   and double exponential distributions.

\clearpage 

\begin{figure}[h]
\centering
\begin{subfigure}{.45\textwidth}\centering
	\includegraphics[width=\textwidth]{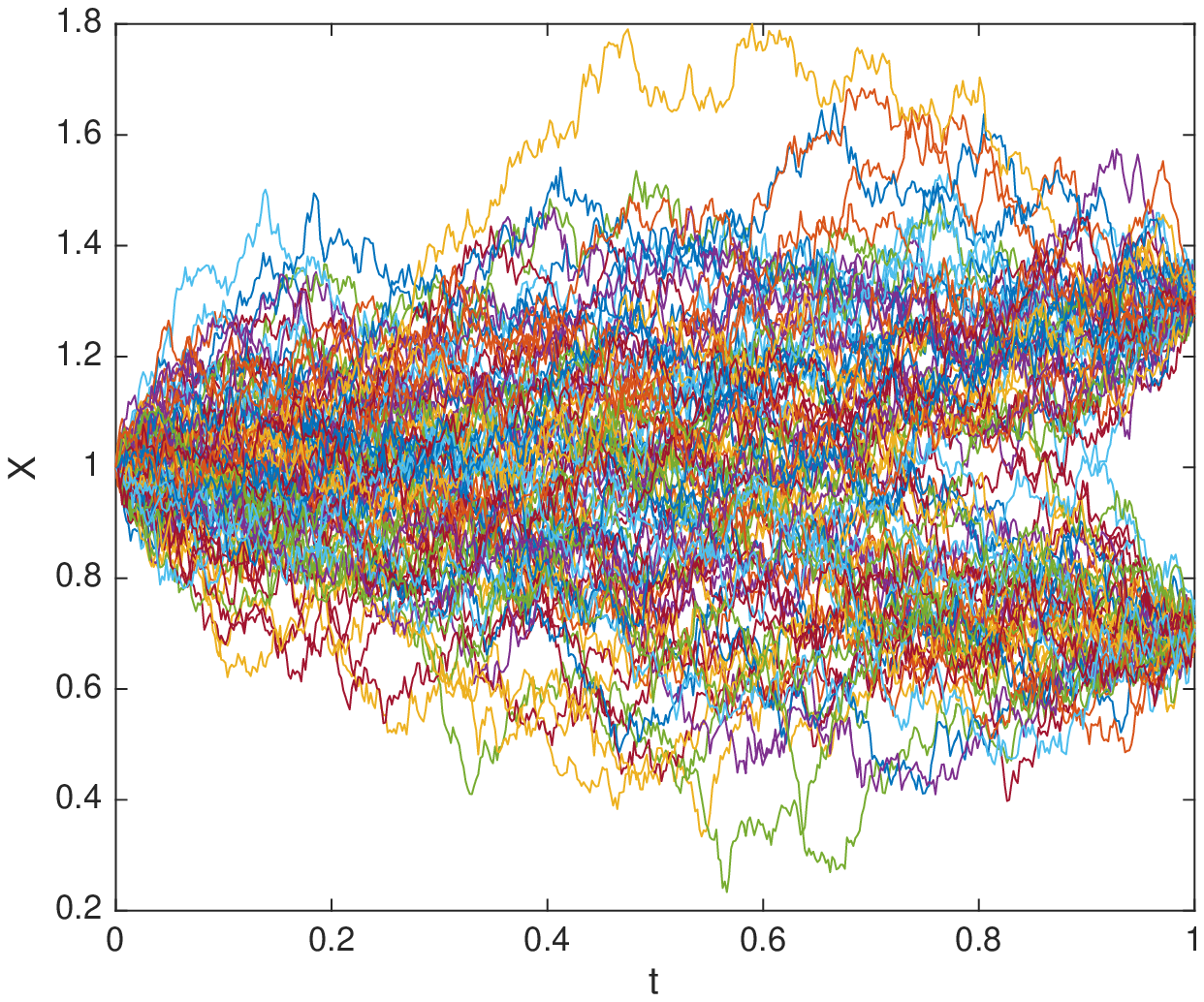}
\caption{}\label{100paths_two_point}
\end{subfigure}\hfill
\begin{subfigure}{.45\textwidth}\centering
	\includegraphics[width=\textwidth]{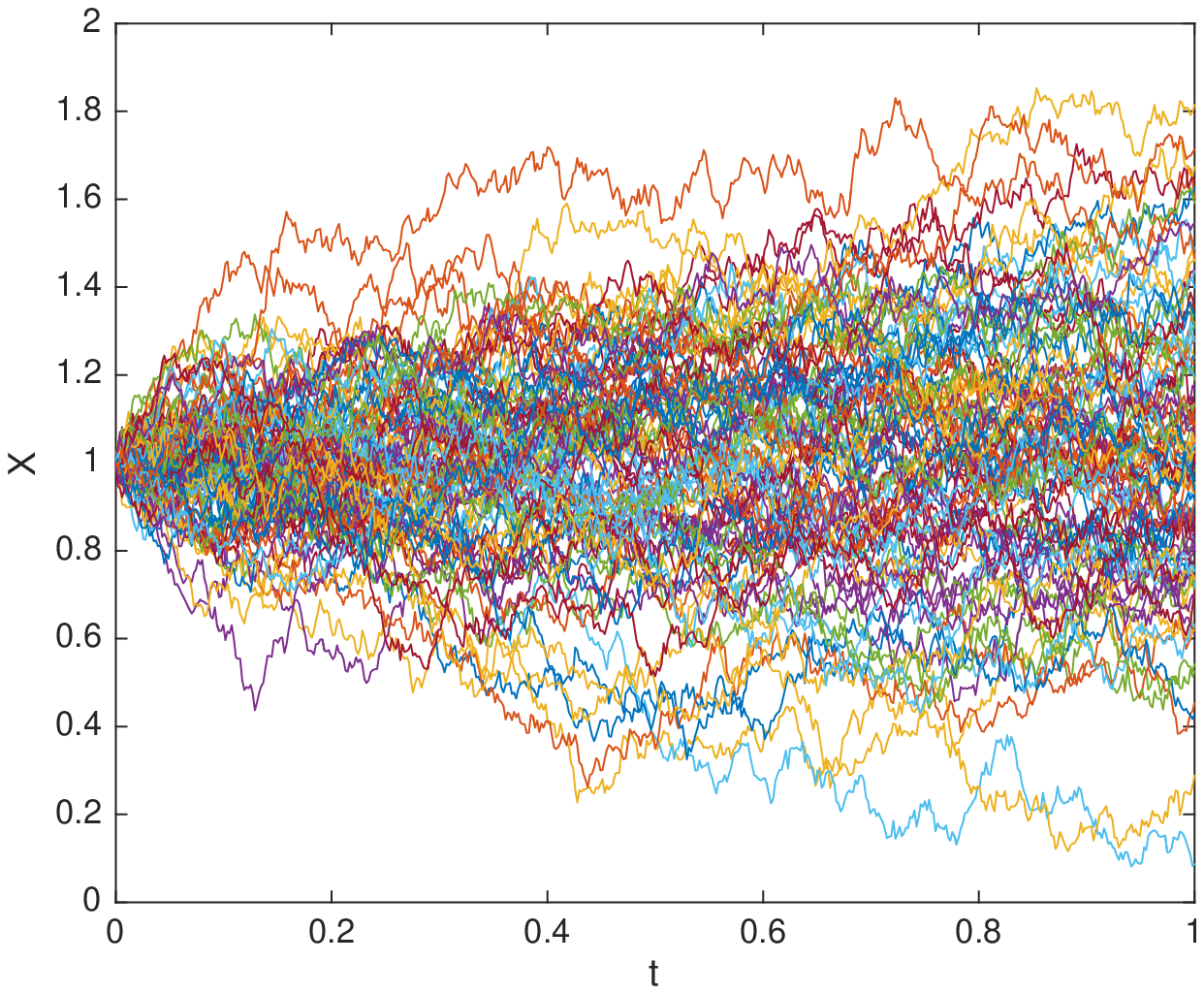}
\caption{}\label{100paths_normal}
\end{subfigure}  
\caption{\small{Path simulation of \eqref{LogPrice} with the prior belief on the future log-price following (a) two-point discrete distribution and (b) normal distribution. Parameters: (a) $\delta _u = 0.3, \delta _d = -0.3, p_u = 0.5, p_d = 0.5$; (b) $\mu = 0, \sigma _D = 0.3$. Common parameters: $X_0 = 1, T = 1, \sigma = 0.4$.}}\label{x_vs_t_sde}
\end{figure}

\vspace{40pt}

\begin{figure}[h]
\centering
\begin{subfigure}{.45\textwidth}\centering
	\includegraphics[width=\textwidth]{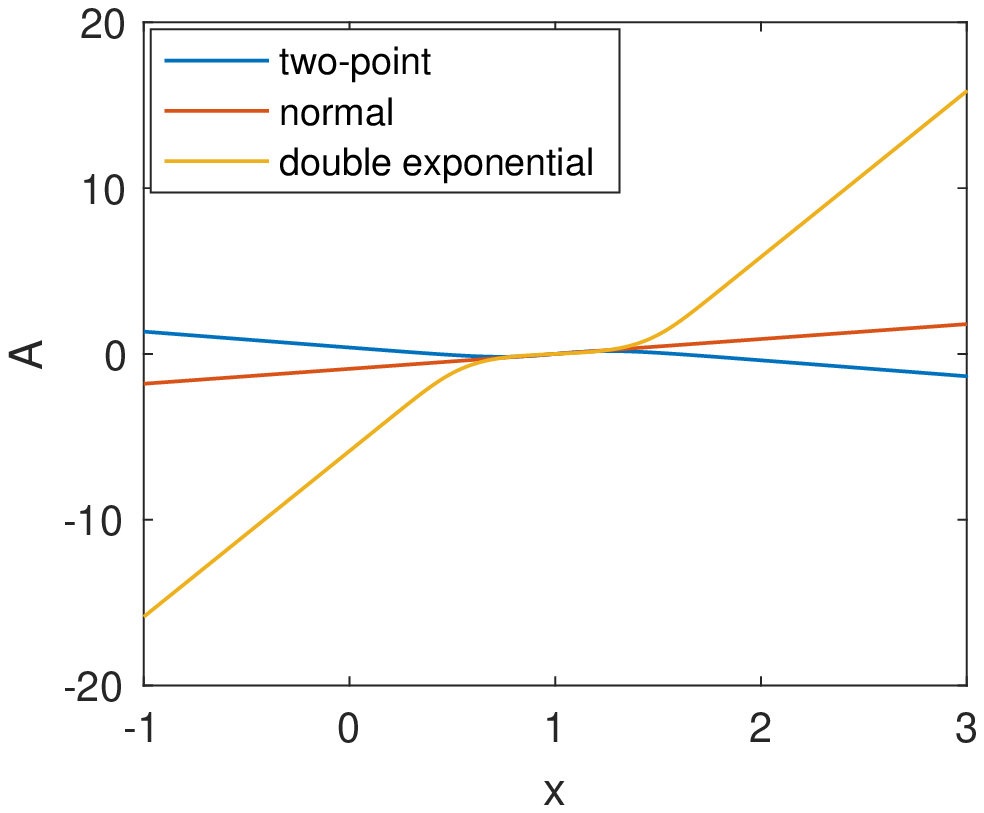}
\caption{}\label{A_t_01}
\end{subfigure} \hfill
\begin{subfigure}{.45\textwidth}\centering
	\includegraphics[width=\textwidth]{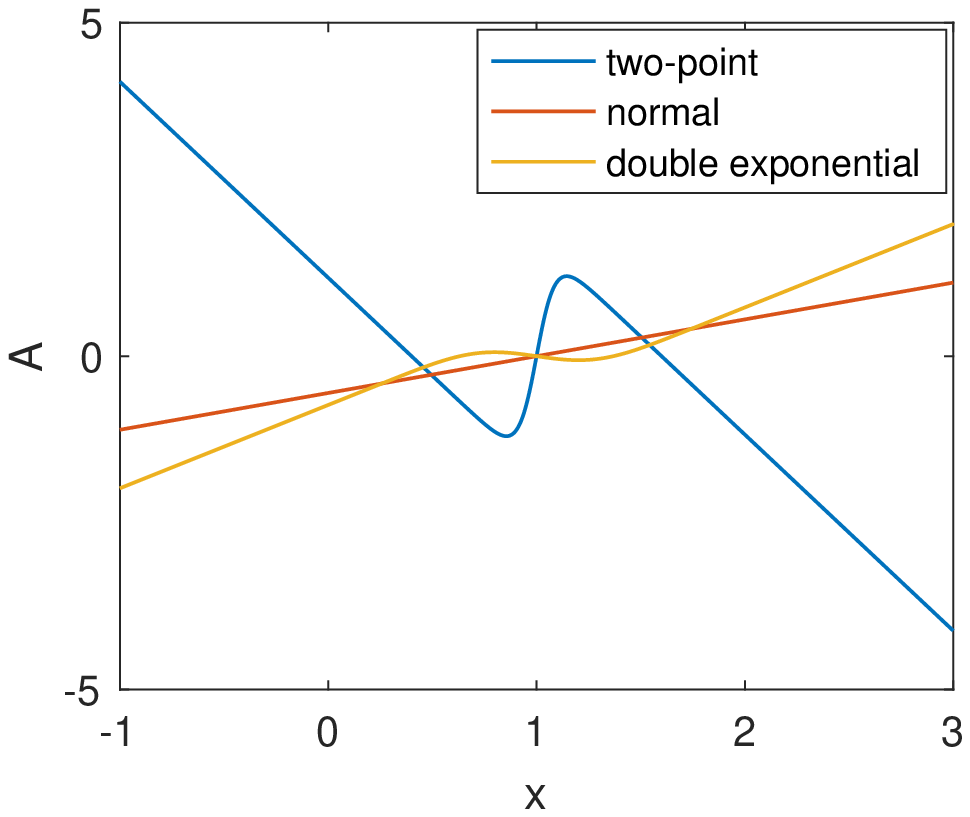}
\caption{}\label{A_t_08}
\end{subfigure}  
\caption{\small{$A(t,x)$ under three different distribution with common \textsc{Mean} = $0$ and \textsc{Var} = $0.36$ at $t = 0.1$ (a) and $t = 0.8$ (b). Parameters: (two-point) $\delta _u = -\delta _d = 0.6, p _u = p _d =0.5$; (normal) $\mu = 0, \sigma _D = 0.6$; (double exponential) $\theta = 0, p_1 = p_2 = 0.5, \lambda _1 = \lambda _2 = 2.357$. Common parameters: $S_0 = 2.72 (X_0 = 1), r=0.1, \T =1, T = 1.1, \sigma = 0.4.$}}\label{A_compare}
\end{figure}
\clearpage

\section{Optimal Liquidation Problems}\label{sect-optimalproblem}
With the asset  price dynamics given above, we now consider   a trader who holds  the underlying asset $S$ or an option written on $S$, and seeks to maximize the expected value from selling the security. Let $f(t,x) \in C^2 ([0, \infty) \times \R)$ be a generic reward function, representing the value received from the security sale at time $t$ at log-price $x$. We assume  a constant interest rate $r>0$, which is also the discount rate used by the trader. 

In order to determine the optimal timing to sell, the trader solves the optimal stopping problem
\begin{align}
V(t, x) = \sup_{\tau \in\setT_{t,\bar{T}}}{\E}\left\{e^{-r(\tau - t)}f(\tau, X_\tau) | X_t = x \right\}, \label{OptimalProb}
\end{align}
where $\setT_{t,\bar{T}}$  is the set of all stopping times with respect to $\mathbb{F}$ taking values between $t$ and $\T$, with $0\le t\le \T \le T$. Here, $\bar{T}$ is the trading deadline which can come before the expiration date of the option $T$.  For all securities considered herein, the  associated reward function $f(t,x)$ is defined through $T$. 

 This problem can be represented in an alternative probabilistic form. To this end, we  first define  the process 
\begin{align}
Y_t = e^{-rt} f(t, X_t), \quad 0\le t \le T.\label{Y}
\end{align}
By  \eqref{LogPrice} and Ito's formula,  we  obtain the SDE
\begin{align}
dY_t 
	&= e^{-rt}\left(-rf(t,X_t) + f_t (t, X_t) + f_x (t, X_t) A(t, X_t) + \frac{\sigma^2}{2}f_{xx} (t, X_t)\right)dt + e^{-rt}f_x (t, X_t)\sigma dW_t \nonumber\\
	&= e^{-rt} G(t, X_t) dt +  e^{-rt} f_x (t, X_t) \sigma dW_t, \label{Ito2}
\end{align}
where we have denoted $f_t\equiv\frac{\partial{f}}{{\partial t}}$, $f_x\equiv\frac{\partial{f}}{{\partial x}}$, $f_{xx}\equiv\frac{\partial ^2 f}{{\partial x^2}}$, and 
\begin{align}
G(t, x) := -rf(t, x) + f_t (t, x) + f_x (t, x) A(t, x) + \frac{\sigma^2}{2} f_{xx} (t, x). \label{G}
\end{align}
The function $G(t,x)$, called the drive function (see \cite{LeungShirai} for the terminology), determines the sign of the drift of the SDE for the discounted reward process  $Y_t = e^{-rt} f(t, X_t)$.  

Integrating \eqref{Ito2} and substituting in \eqref{OptimalProb}, the value function can be expressed as
\begin{align}
V(t, x) =  \sup_{\tau \in\setT_{t,\bar{T}}}{\E} \left\{\int_t^\tau e^{-r(u-t)}G(u, X_u)du | X_t = x \right\} + f(t, x). \label{VG}
\end{align}
Rearranging the terms in \eqref{VG}, we define the difference between the value function $V(t, x)$ and the reward function $f(t, x)$  to be the delayed liquidation premium, i.e.
 \begin{align}
L (t, x) &:= V(t, x) - f(t, x)\notag \\
& =\sup_{\tau \in\setT_{t,\bar{T}}}{\E} \left\{\int_t^\tau e^{-r(u-t)}G(u, X_u)du | X_t = x \right\}.
\label{L}
\end{align}
For every position held, there is an embedded timing option to sell. The delayed liquidation premium quantifies the value of optimally waiting to exercise this timing option. Since $V(t,x)\ge f(t,x)$ for all $(t,x)$, this follows from \eqref{L} that  $L(t,x)\ge 0$ for all $(t,x)$, meaning that the delayed liquidation premium is always positive.    

By standard optimal stopping theory \citep[Theorem D.12]{KaratzasShreve01}, the optimal liquidation time, associated with $V(t,x)$ or $L(t,x)$, is given by
\begin{align}
\tau^* &= \inf\{ \,u \in [t,\bar{T}]\,:\, V(u,X_u) =f(u,X_u)\, \}\notag\\
&= \inf\{ \,u \in [t,\bar{T}]\,:\, L(u,X_u) =0\, \}.\label{Ltau}
\end{align}
In other words, it is optimal for the trader to exercise the timing option and close out the position  as soon as the optimal liquidation premium $L$ vanishes. Accordingly, the trader's optimal liquidation strategy can be described by the exercise region $\mathcal S$ and continuation region $\mathcal D$, namely,
\begin{align}
\mathcal{S} &=\{(t, e^x)\in [0,\bar{T}]\times \R_+ : \ L(t,x)=0\},\\
\mathcal{D}  &=\{(t, e^x)\in [0,\bar{T}]\times \R_+ : \ L(t,x)>0\}.
\end{align}

 On the other hand, if the delayed liquidation premium is always strictly positive, then the trader finds it optimal to wait through the end of the trading horizon. In particular, we may have $L(\bar{T},x)>0$ for all $x$, so we interpret $\tau^* = \T$ as never exercising at all. Therefore,    we can now  identify the conditions under which it is optimal to immediately liquidate, or hold the asset/option position through $\bar{T}$.

\begin{proposition} \label{prop_G}
 Let $t\in[0,\T]$ be the current time. Then, we have
\begin{enumerate}\item $G(u, x) > 0$,  $\forall(u, x)\in [t,\T]\times \mathbb{R}_+ \implies \tau^* = \T$. 
\item  $G(u, x) \leq 0$, $ \forall(u, x)\in [t,\T]\times \mathbb{R}_+ \implies \tau^* = t$.
\end{enumerate}\end{proposition}
\begin{proof} From \eqref{VG} that if $G(u, x)$ is positive (resp. negative) $\forall (u,x) \in [t, \T] \times \R ^+$, then the discounted reward function is maximized at the longest (resp. shortest) stopping time, i.e. $\tau^* = \T$ (resp. $\tau^* = t$).
\end{proof}

In addition to the two extremal cases, we can also address other cases and solve for the trader's nontrivial trading strategies.  To do so, let us write down the variational inequality associated with the value function $V(t,x)$ (see \eqref{OptimalProb}) with a general reward function $f(t,x)$. First define the differential operator
 \begin{align}
 \L\{\cdot\} := -r \cdot + \frac{\partial\cdot}{\partial t} + A(t,x)\frac{\partial\cdot}{\partial x} + \frac{\sigma^2}{2}\frac{\partial^2 \cdot}{\partial x^2}. \label{Lop}
\end{align}
The optimal stopping problem  $V$ is solved from  the variational inequality:
\begin{align}
\textrm{max}\left\{\,\L V(t,x)\,, f(t,x) - V(t,x)\,\right\} &= 0,\label{VIV}
\end{align}
where $(t, x) \in [0,\T) \times \R$, with the terminal condition $V(\bar{T},x) = f(\bar{T},x)$ for all $x\in \R$. We refer to \citep[Section 7]{LeungShirai} for a detailed proof for the existence and uniqueness  of a strong solution to variational inequalities of this form. A more comprehensive reference is \cite{Bensoussan}. We will discuss in Section \ref{sect-finitedifference} our numerical scheme to solve for the optimal trading strategies.

Next, we examine the trading strategies for stock and options, and study the varying effects of the trader's belief encoded in the random variable $D$. For each security type, we will derive  the corresponding drive function. It will then be the inputs for the variational inequality \eqref{VIV}, which will be solved numerically for the optimal trading strategies.  We consider   several  combinations in securities   and beliefs to see any major differences  in the strategies.

 \subsection{Stocks}
For selling the stock $S$, the reward function is simply $f(x) = e^x$. Then, the drive function, denoted by $G_{stock} (t, x)$,  for selling the stock can be computed via \eqref{G}. Precisely, we get 
\begin{align}
G_{stock} (t, x) = (-r + A(t, x) + \frac{1}{2} \sigma ^2) e^x, \label{Gstock}
\end{align}for $(t,x) \in [0,T]\times \R$. As seen in Figure \ref{A_compare}, the function $A(t,x)$ and thus drive function depend heavily on the prescribed distribution of $D$, and may be nonlinear. In general, it is difficult to pinpoint the behavior of the drive function $G_{stock}(t,x)$. Nevertheless, under the normal distribution for $D$, we obtain the following properties for $G_{stock}$ in $x$ and in $t$ respectively.

\begin{proposition}\label{normalx}
Suppose that  the trader's belief on log-price follows a normal distribution as in Example \ref{example2}. Then, the drive function $G_{stock} (t, x)$ given  in \eqref{Gstock} is
\begin{enumerate}[(i)]
\item downward-sloping and concave in $x$ if 
\[\sigma _D > \sqrt{T} \sigma   \text{ and }   x \leq q(t) - 2,\] 
or 
\[\sigma _D < \sqrt{T} \sigma   \text{ and }   x \geq q(t) - 1,\]

\item downward-sloping and convex in $x$ if 
\[\sigma _D > \sqrt{T} \sigma   \text{ and }  q(t) - 2 \leq x \leq q(t) - 1,\]

\item upward-sloping and concave in $x$ if 
\[\sigma _D < \sqrt{T} \sigma   \text{ and }  q(t) - 2 \leq x \leq q(t) - 1,\] 

\item upward-sloping and convex in $x$ if 
\[\sigma _D > \sqrt{T} \sigma   \text{ and }   x \geq q(t) - 1,\] 
or 
\[\sigma _D < \sqrt{T} \sigma   \text{ and }   x \leq q(t) - 2,\]

\end{enumerate}
where \[q(t) := X_0  - \frac{\mu \sigma ^2 T + ( \frac{1}{2} \sigma^2-r)(t \sigma _D ^2 + T \sigma ^2 (T - t))}{\sigma _D ^2 - T \sigma ^2}.\]
\end{proposition}


\begin{proposition}\label{normalt}
Suppose that  the trader's belief on log-price follows a normal distribution as in Example \ref{example2}. Then,  the drive function in \eqref{Gstock} is
\begin{enumerate}[(i)]
\item downward-sloping and concave in $t$ if 
\[\sigma _D < \sqrt{T} \sigma   \text{ and }   x \geq X_0 - \frac{\mu \sigma ^2 T}{\sigma _D ^2 - T \sigma ^2}   ,\] 

\item downward-sloping and convex in $t$ if 
\[\sigma _D > \sqrt{T} \sigma   \text{ and }  x \geq X_0 - \frac{\mu \sigma ^2 T}{\sigma _D ^2 - T \sigma ^2}, \]

\item upward-sloping and concave in $t$ if 
\[\sigma _D > \sqrt{T} \sigma   \text{ and }   x \leq X_0 - \frac{\mu \sigma ^2 T}{\sigma _D ^2 - T \sigma ^2},\] 

\item upward-sloping and convex in $t$ if 
\[\sigma _D < \sqrt{T} \sigma   \text{ and }   x \leq X_0 - \frac{\mu \sigma ^2 T}{\sigma _D ^2 - T \sigma ^2}.\] 
\end{enumerate}
\end{proposition}
The proof is provided in Section \ref{Proof_proposition5}. In addition, if $\sigma _D = \sqrt{T} \sigma$, then the drive function is independent of $t$. Under the normal distribution for $D$, $G_{stock} (t,x)$  is upward-slopping and convex if $r < \frac{\mu}{T} + \frac{1}{2}\sigma ^2$, and  downward-slopping and concave if $r  > \frac{\mu}{T} + \frac{1}{2}\sigma ^2$. However, for two-point discrete distribution and double exponential distribution, the properties of the drive function are not explicitly available, but the function can be numerically computed instantly using \eqref{Gstock}.

Figure \ref{boundaries} presents the optimal trading strategies for selling a stock with beliefs corresponding to the  two-point discrete distribution, normal distribution and double exponential distribution.  Typically, as seen in panels (a), (b) and (d), we observe a decreasing boundary over time. This means that the trader tends to sell the stock when the stock price is sufficiently high, but is willing to liquidate at a lower price as the trading deadline approaches. However, Figure \ref{boundary_CE_bernoulli}  shows that the optimal boundary is increasing in time under the two-point discrete distribution for $D$ (see Example \ref{example1}). Note that in Figure \ref{boundary_CE_bernoulli}, the trader has a more divergent belief in the sense that the terminal stock price will end up either very high or very low.  Under our model, this suggests that the trader will hold the stock if the price is high since the trader believes the  price will likely increase further, and the trader will sell it if the price is low because  the price is believed to go even lower under the trader's belief.   Figures \ref{boundary_EC_bernoulli} and \ref{boundary_CE_bernoulli} represent the trader's two contrasting  beliefs. The magnitude of the parameter  $\delta$ guides the trader's reaction to new price information and ultimately the trading strategy. This highlights that the behavior of the boundary can be significantly changed by the choice of distribution of $D$ and associated parameters.

In Figure \ref{para_sensitivity} we illustrate the sensitivities of the optimal strategies for selling stock with respect to a number of parameters under three different distributions for $D$. As expected, the boundary increases as the mean (left panels) of $D$ increases and variance decreases (right panels).


Due to the nonlinearity of the drive functions under certain distributions, the trader may need to adopt more complicated trading strategies. As  we have  shown in Figures \ref{boundaryf_delta08_p05},     if the  initial stock price is in the continuation region, then the trader does not sell  until the price  reaches the   exercise boundaries.  However, the continuation region is sandwiched by two upper and lower disconnected exercise regions. This means that not only when the price is sufficiently high but also sufficiently low, the trader will sell the stock. We also notice another continuation region when the stock price is very low. The intuition surrounding this continuation region is that when the current stock price is so low, the trader speculates the price will  rebound and  chooses to hold onto the stock. 

Figure \ref{boundaryf_la03_p05} suggests a  different trading strategy under the double exponential distribution for $D$. Since the lower continuation region is relatively narrow, for  most starting prices the trader will find it optimal to sell the stock   immediately. It is clear that trader's prior belief impacts her trading strategies. In Figure \ref{boundaryf_la04_p05}, the parameters $\lambda _1$ and $\lambda _2$ are relatively small, and the  mean and variance of $D$ are $0$ and variance $0.16$ (see \eqref{double_mean} and \eqref{double_var}). Compared to the  market view with the same mean but  smaller variance of value  0.09 in Figure \ref{boundaryf_la03_p05},   the upper continuation region is smaller (and vanishes for small $t$) and lower continuation region becomes larger in Figure \ref{boundaryf_la04_p05}.   

If the trader adopts the   normal distribution for $D$ with   mean zero and variance   $\sigma ^2 T$, then she will never update the log-price dynamics because the drift of $X$ is zero, i.e. $A(t,x) = 0$ $\forall (t,x)$. Other than this special degenerate case,  the trader can   take normal priors with different variances. We notice  that the  continuation/exercise regions are disconnected under the two-point discrete distribution and double exponential distribution in Figures \ref{boundaryf_delta08_p05} and   \ref{disconnected}, but  there is a unique exercise boundary under normal distribution.

\begin{figure}[h]
\centering
\begin{subfigure}{.49\textwidth}\centering
	\includegraphics[width=\textwidth]{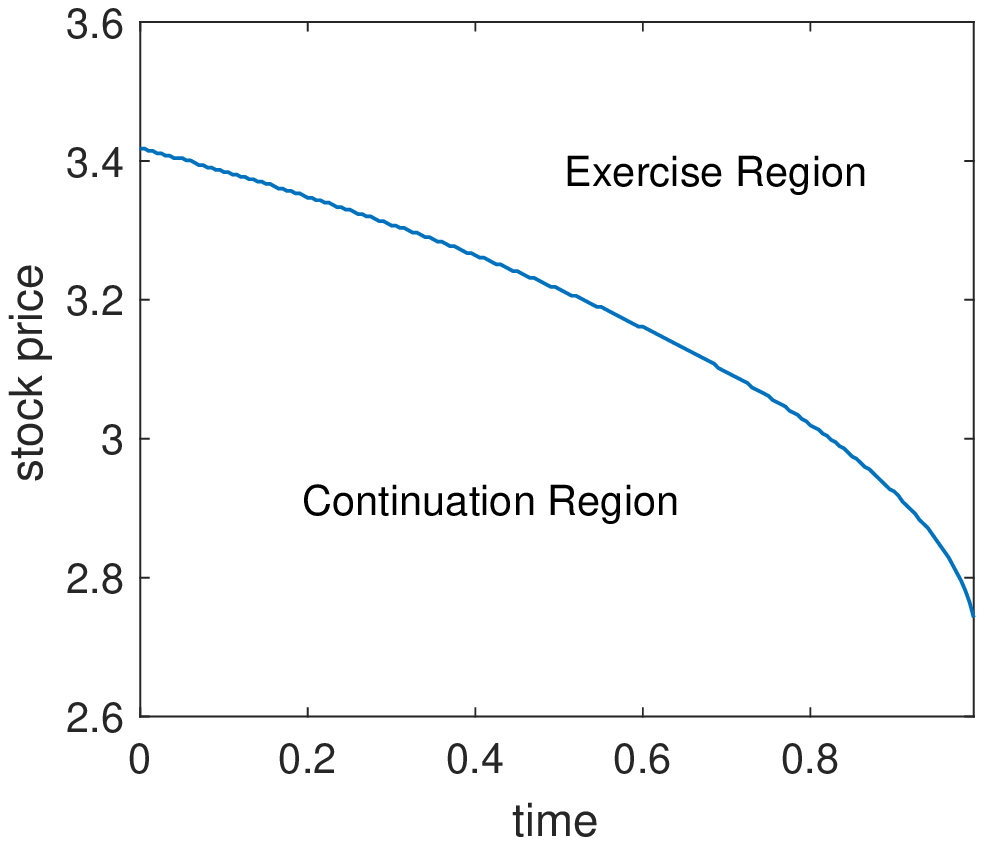}
\caption{}\label{boundary_EC_bernoulli}
\end{subfigure}\hfill
\begin{subfigure}{.49\textwidth}\centering
	\includegraphics[width=\textwidth]{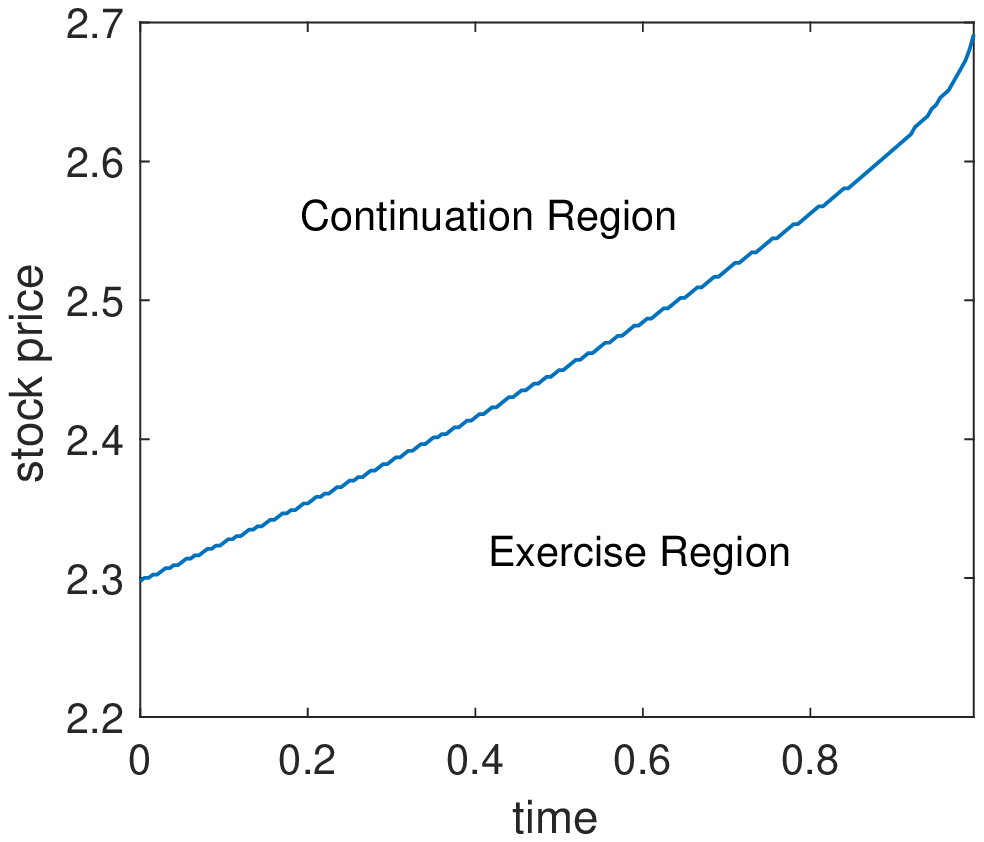}
\caption{}\label{boundary_CE_bernoulli}
\end{subfigure} \hfill
\begin{subfigure}{.49\textwidth}\centering
	\includegraphics[width=\textwidth]{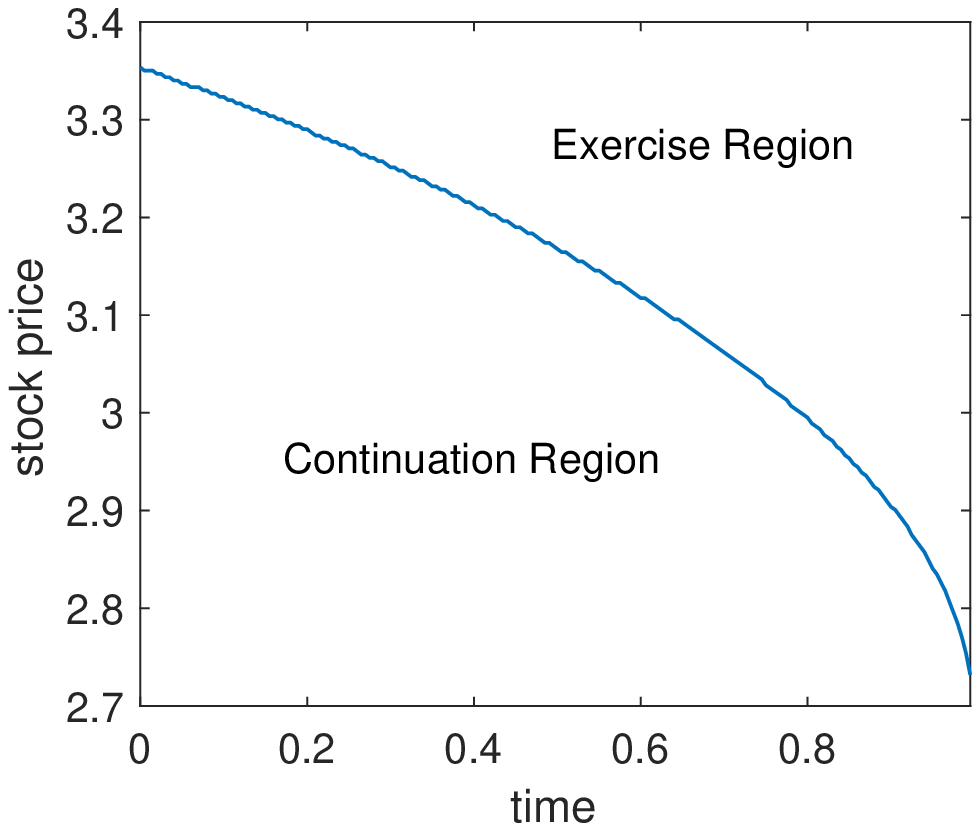}
\caption{}\label{boundary_EC_sigma02}
\end{subfigure}   \hfill
\begin{subfigure}{.49\textwidth}\centering
	\includegraphics[width=\textwidth]{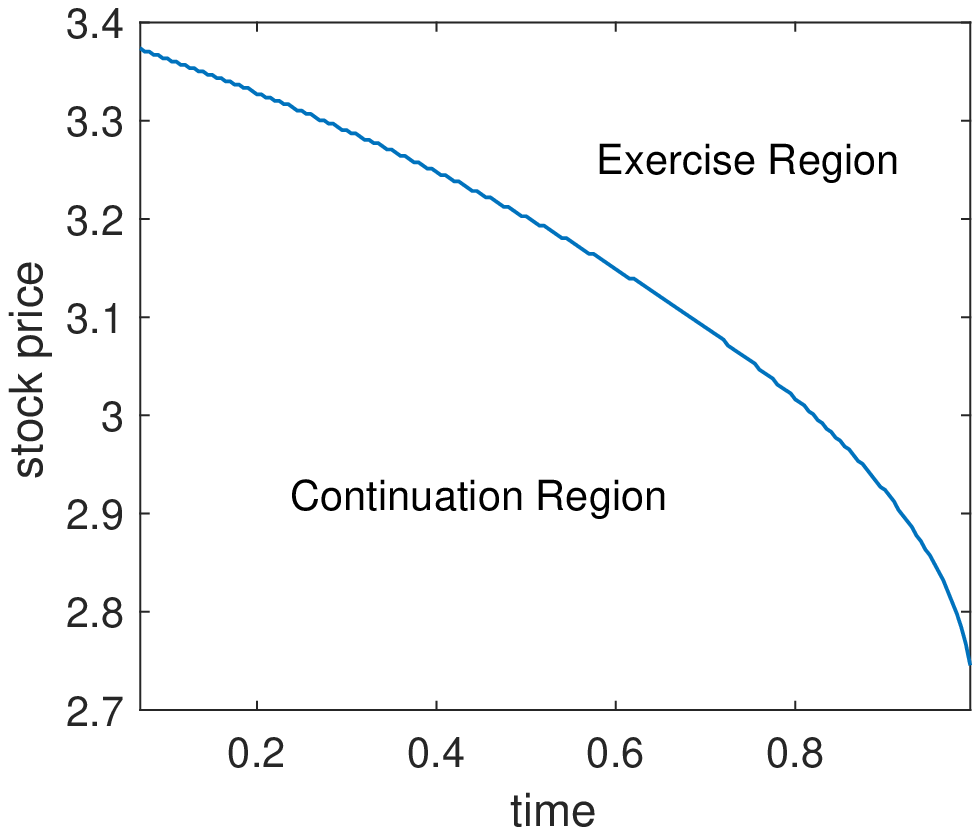}
\caption{}\label{boundary_EC_double}
\end{subfigure} 
\caption{Optimal boundaries for selling stock under two-point discrete distribution (a \& b), normal distribution (c) and double exponential distribution (d). Parameters for each plot as follows: (a) $\delta _u = -\delta _d = 0.1, p _u = p _d =0.5$; (b) $\delta _u = -\delta _d = 2, p _u = p _d =0.5$; (c) $\mu = 0, \sigma _D = 0.2$; (d) $\theta = 0, p_1 = p_2 = 0.5, \lambda _1 = \lambda _2 = 10$. Common parameters: $S_0 = 2.72 (X_0 = 1), r=0.1, \T=1, T = 1.1, \sigma = 0.4.$}\label{boundaries}
\end{figure}

\begin{figure}[h]
\centering
\begin{subfigure}{.47\textwidth}\centering
	\includegraphics[width=\textwidth]{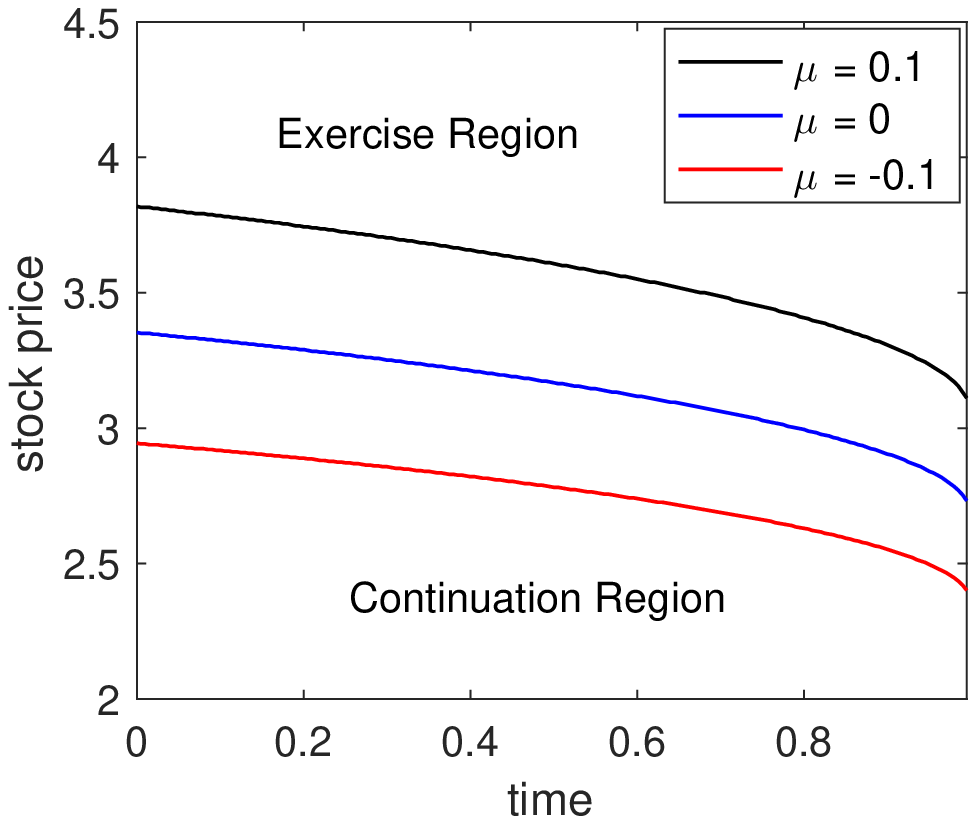}
\caption{}\label{mu_sigmaD_02}
\end{subfigure}\hfill
\begin{subfigure}{.47\textwidth}\centering
	\includegraphics[width=\textwidth]{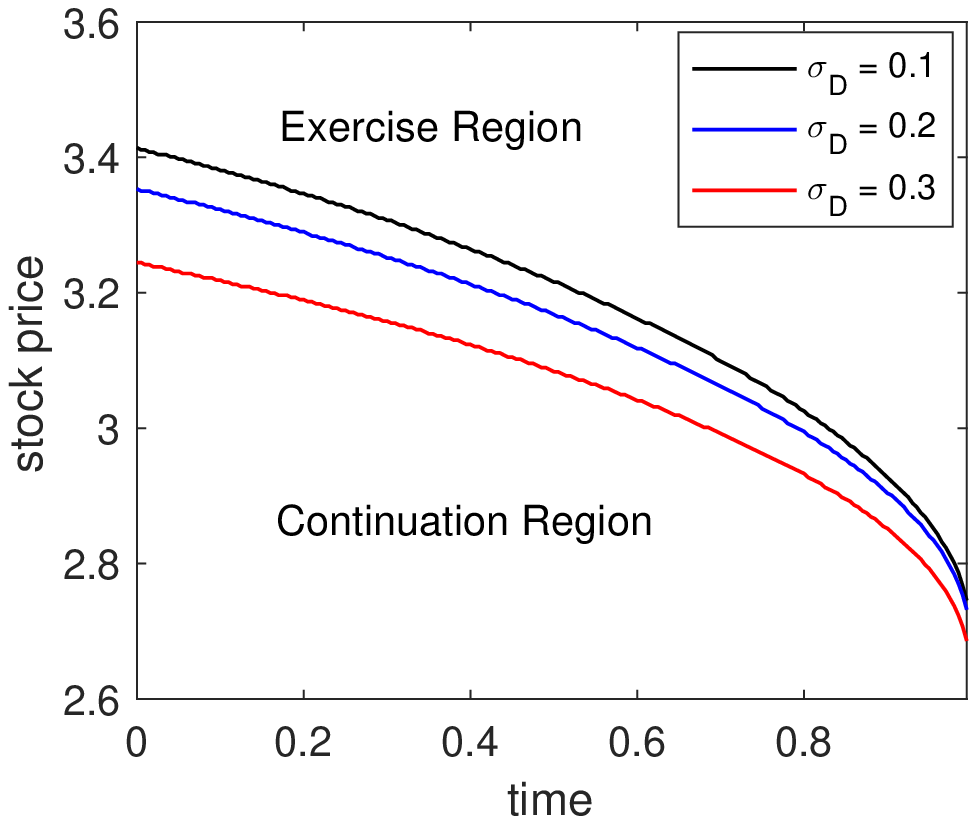}
\caption{}\label{mu_0_sigmaD}
\end{subfigure} \hfill
\begin{subfigure}{.47\textwidth}\centering
	\includegraphics[width=\textwidth]{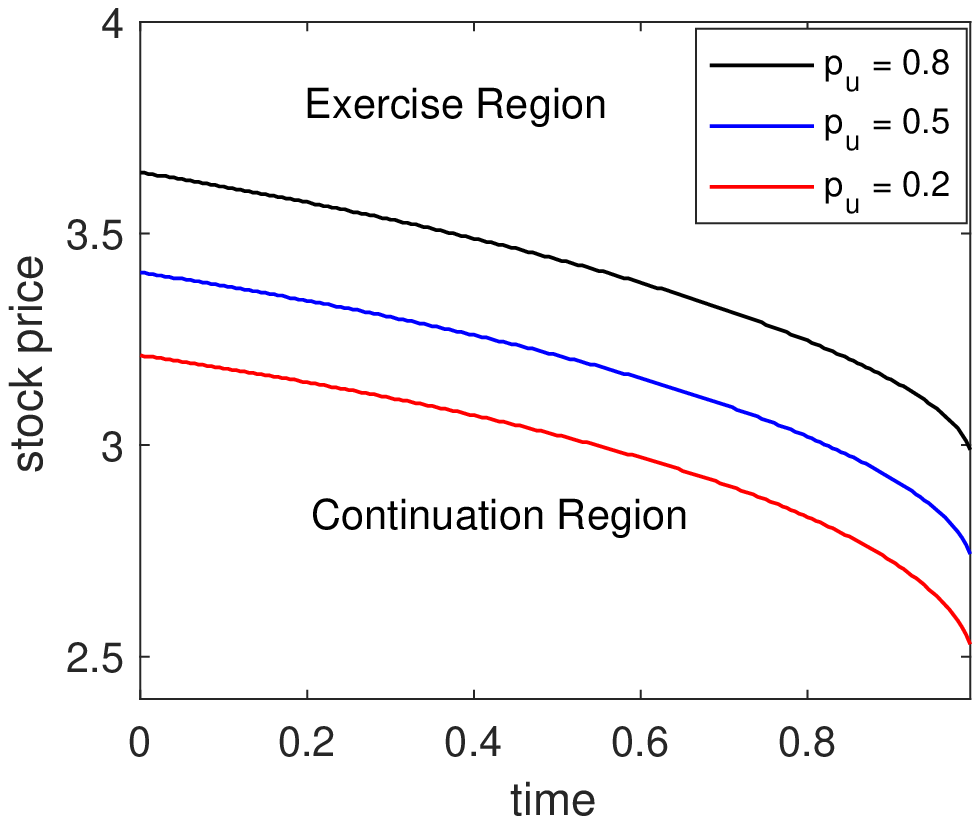}
\caption{}\label{puDu_01Dd_m01}
\end{subfigure}   \hfill
\begin{subfigure}{.47\textwidth}\centering
	\includegraphics[width=\textwidth]{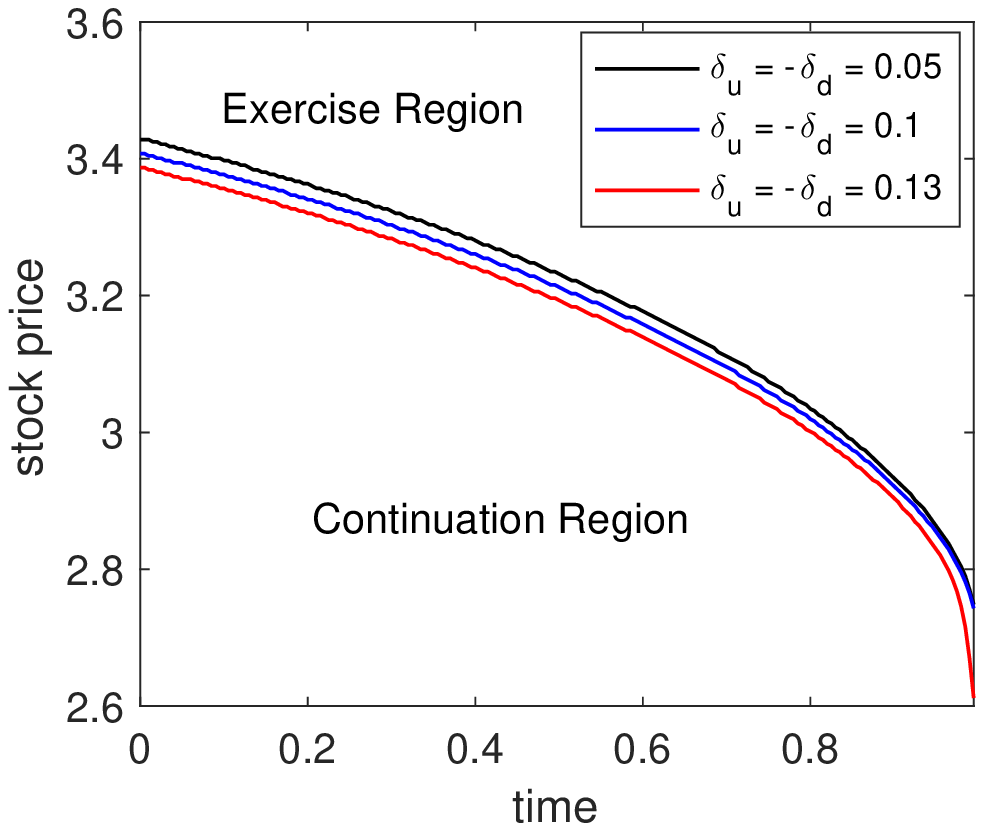}
\caption{}\label{pu_05pd_05D}
\end{subfigure} 
\begin{subfigure}{.47\textwidth}\centering
	\includegraphics[width=\textwidth]{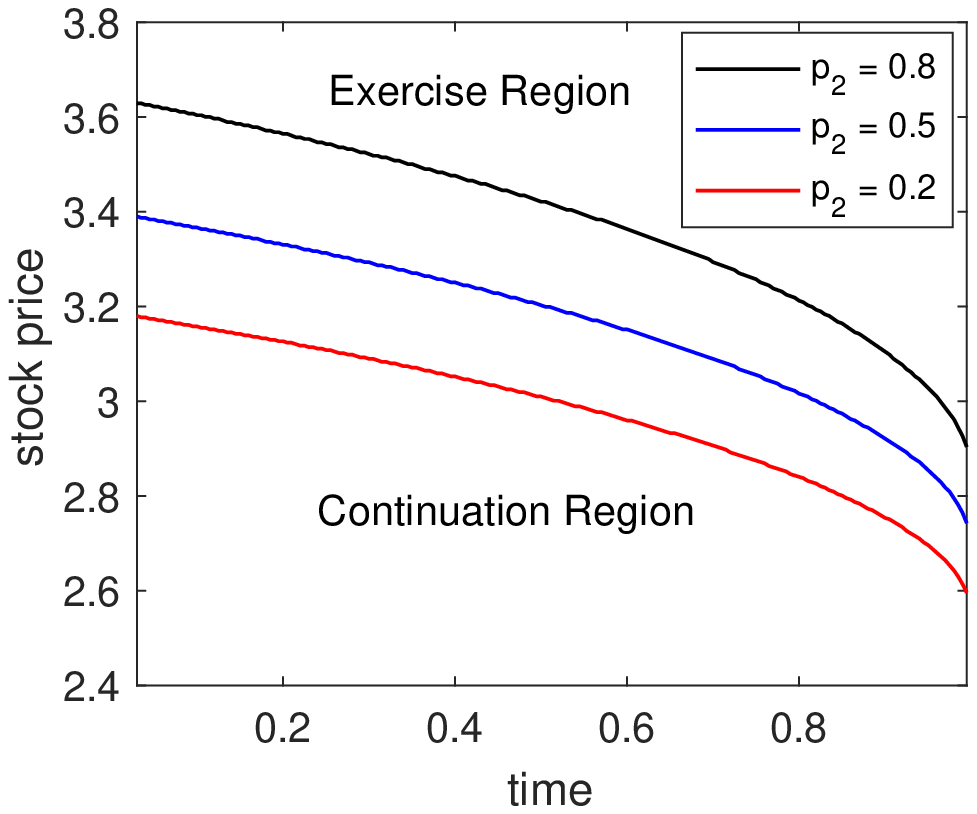}
\caption{}\label{plam1_10lam2_10}
\end{subfigure}\hfill
\begin{subfigure}{.47\textwidth}\centering
	\includegraphics[width=\textwidth]{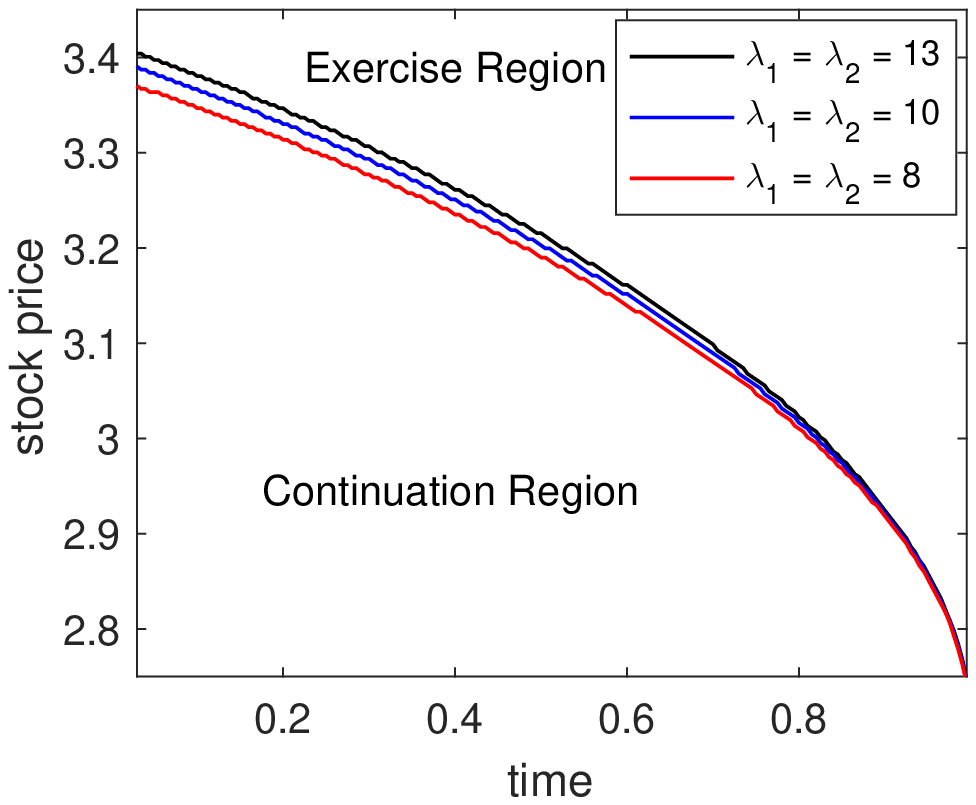}
\caption{}\label{p1_05p2_05lam}
\end{subfigure} \hfill
\caption{Optimal trading boundaries for  stock liquidation under the two-point discrete distribution (a \& b), normal distribution (c \& d) and double exponential distribution (e \& f). Parameters: (a) $\sigma _D = 0.2$; (b) $\mu = 0$; (c) $\delta _u = -\delta _d = 0.1$; (d) $p _u = p _d =0.5$; (e) $\lambda _1 = \lambda _2 = 10$; (f) $p_1 = p_2 = 0.5$. Other common parameters are same to Figure \ref{boundaries}.}\label{para_sensitivity}
\end{figure}

\begin{figure}[h]
\centering
	\includegraphics[width=3in]{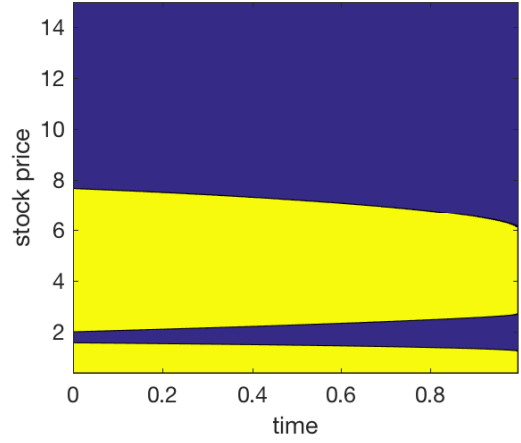}
\caption{Optimal trading regions for  stock liquidation under the two-point discrete distribution. The continuation regions are in yellow (light) color and the exercise regions are in blue (dark), and they are disconnected. Parameters: $\delta _u = -\delta _ d = 0.8, p_u = p_d = 0.5$. Common parameters: $X_0=1, r=0.1, \T =1, T = 1.1, \sigma = 0.4$. } \label{boundaryf_delta08_p05}
\end{figure}

\begin{figure}[h]
\centering
\begin{subfigure}{.49\textwidth}\centering
	\includegraphics[width=\textwidth]{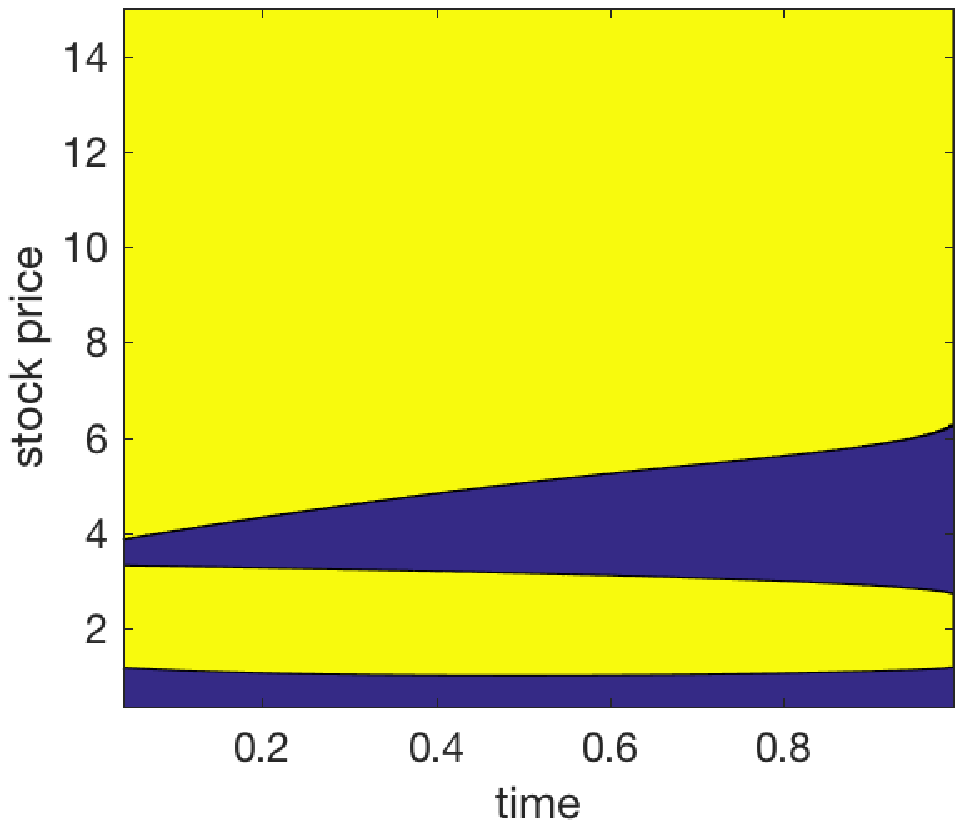}
\caption{}\label{boundaryf_la03_p05}
\end{subfigure}\hfill
\begin{subfigure}{.49\textwidth}\centering
	\includegraphics[width=\textwidth]{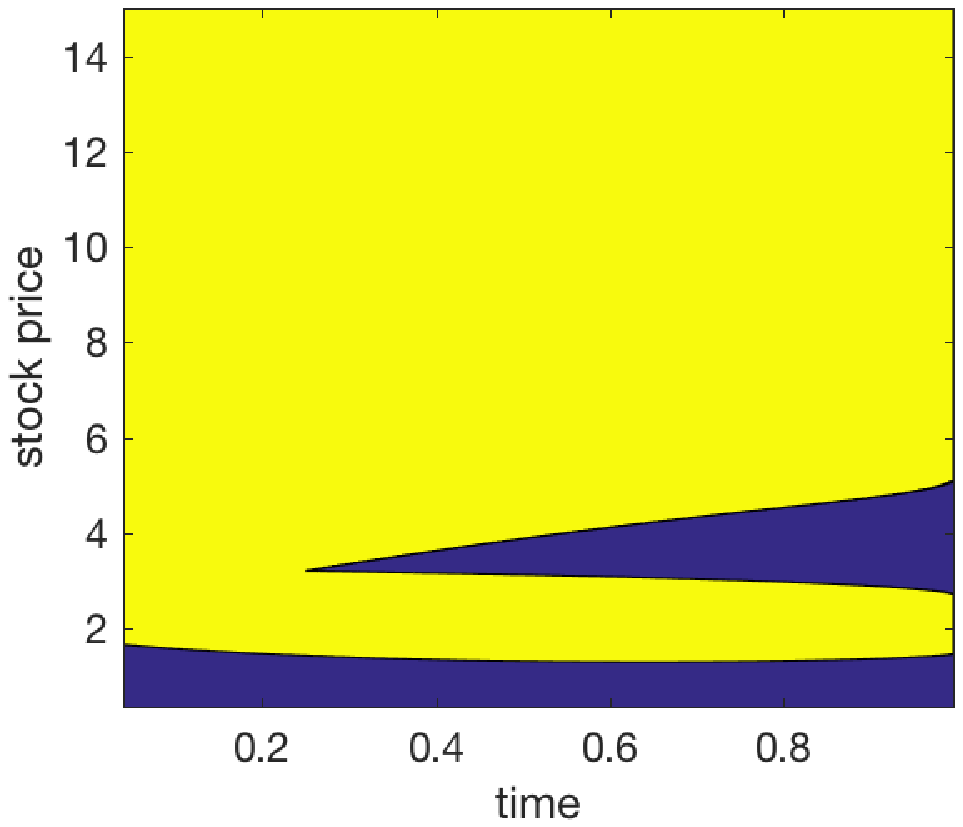}
\caption{}\label{boundaryf_la04_p05}
\end{subfigure} \hfill
\caption{Optimal trading regions for  stock liquidation under the double exponential  distribution. The continuation regions are in yellow (light) color and the exercise regions are in blue (dark), and they are disconnected.  Parameters for (a): $\theta = 0, \lambda _1 = \lambda _2 = 4.714, p_1 = p_2 = 0.5$; parameters for (b): $\theta = 0, \lambda _1 = \lambda _2 = 3.536, p_1 = p_2 = 0.5$. Common parameters: $X_0 = 1, r=0.1, \bar{T} =1, T = 1.1, \sigma = 0.4.$. }\label{disconnected}
\end{figure}

\clearpage

\subsection{Call and Put Options}
 
Given that the stock price follows SDE \eqref{LogPrice} under the historical measure $\mathbb{P}$, the risk-neutral counterpart follows the geometric Brownian motion 
\begin{align}
d S_t = rS_t \,dt +  \sigma S_t  \,dW^{\mathbb{Q}}_t, \end{align}
where $W^{\mathbb{Q}}_t$ is a standard Brownian motion under the risk-neutral measure $\mathbb{Q}$.  Therefore, the no-arbitrage prices of European call and put options are found from the Black-Scholes pricing formulae.  

The price of a European  call with strike price $K$ and expiration date $T$ is given by
\begin{equation}
C_{BS}(t,e^{x}) = \Phi (d_1) e^{x} - \Phi (d_2) Ke^{-r(T-t)}, \label{C_bs}
\end{equation}
for $(t,x) \in [0,T]\times \R$, where 
\begin{align}
d_1 &= \frac{1}{\sigma\sqrt{T-t}}\left(\text{ln}(e^{x}/K) + (r+\sigma^2/2)(T-t)\right), \label{d1} \\
d_2 &= d_1 - \sigma\sqrt{T-t}. \label{d2}
\end{align}
Substituting the above into \eqref{G}, we obtain the drive function

\begin{align}
G_{call}(t, x) &= -r C_{BS}(t, e^x) + \frac{\partial C_{BS} (t, e^x)}{\partial t} + \frac{\partial C_{BS} (t, e^x)}{\partial x} A(t, x) + \frac{\sigma ^2}{2} \frac{\partial ^2 C_{BS}(t, e^x)}{\partial x^2} \nonumber\\
 &= \left(-r + A(t,x) + \frac{\sigma ^2}{2} \right) e^x\Phi (d_1), \label{drive_call}
\end{align}
where $A(t, x)$ for the two-point, normal, double exponential cases are respectively computed by \eqref{A_twopoint}, \eqref{A_normal} and \eqref{A_double}. 
 

In  Figure \ref{boundary_call_sigma01_K100}, we illustrate  the optimal trading strategies for selling an European call under the normal distribution for $D$. As we can see, the trader tends to sell the call option when the stock price is high, but is willing to sell at a lower price as time approaches maturity. The green (light) line shows where  $G_{call} = 0$, and  the continuation region contains $G(t, x)\geq 0$. Recall that  $G(t,x)$ is the integrand in the expression for $L(t,x)$. Naturally, if $G(t, x) > 0$,  it is optimal for the trader to continue to hold   the call  because  she can obtain positive infinitesimal premium by waiting for an infinitesimally small amount of time.  The same argument also applies to the cases with a stock and a put option.

 Next, we consider the liquidation of a European put option.  The Black-Scholes put price is given by 
 \begin{align}
P_{BS}(t, e^{x}) = \Phi (-d_2)Ke^{-r(T-t)} - \Phi (-d_1)e^{x}.
\end{align}
where $K$ is the strike price, $T$ is the expiration date, and $d_1$ and $d_2$ are given by \eqref{d1} and \eqref{d2}.  Using \eqref{G}, the  corresponding drive function $G_{put}$ is
\begin{align}
G_{put}(t, x) = \left(r - A(t,x) - \frac{\sigma ^2}{2} \right) e^x\Phi (-d_1). \label{drive_put}
\end{align}
Figure \ref{boundary_put_sigma01_K100} presents the optimal timing strategies for selling an European put under normal distribution. The figure shows that the trader tends to sell the put option when the stock price is low (i.e. when the put option price is high), but is willing  to sell at a higher stock price (i.e. lower option price) as time approaches maturity. As a result, the exercise region expands and continuation region shrinks as time progresses.

\begin{figure}[h]
\centering
\begin{subfigure}{.49\textwidth}\centering
	\includegraphics[width=\textwidth]{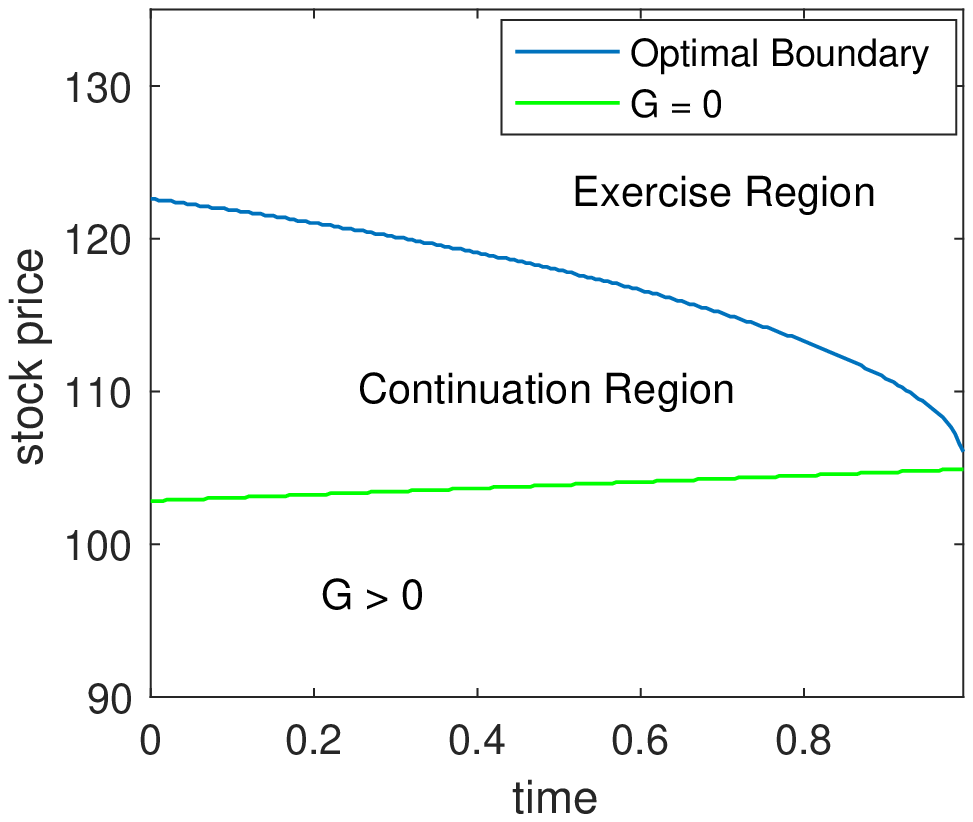}
\caption{}\label{boundary_call_sigma01_K100}
\end{subfigure}\hfill
\begin{subfigure}{.49\textwidth}\centering
	\includegraphics[width=\textwidth]{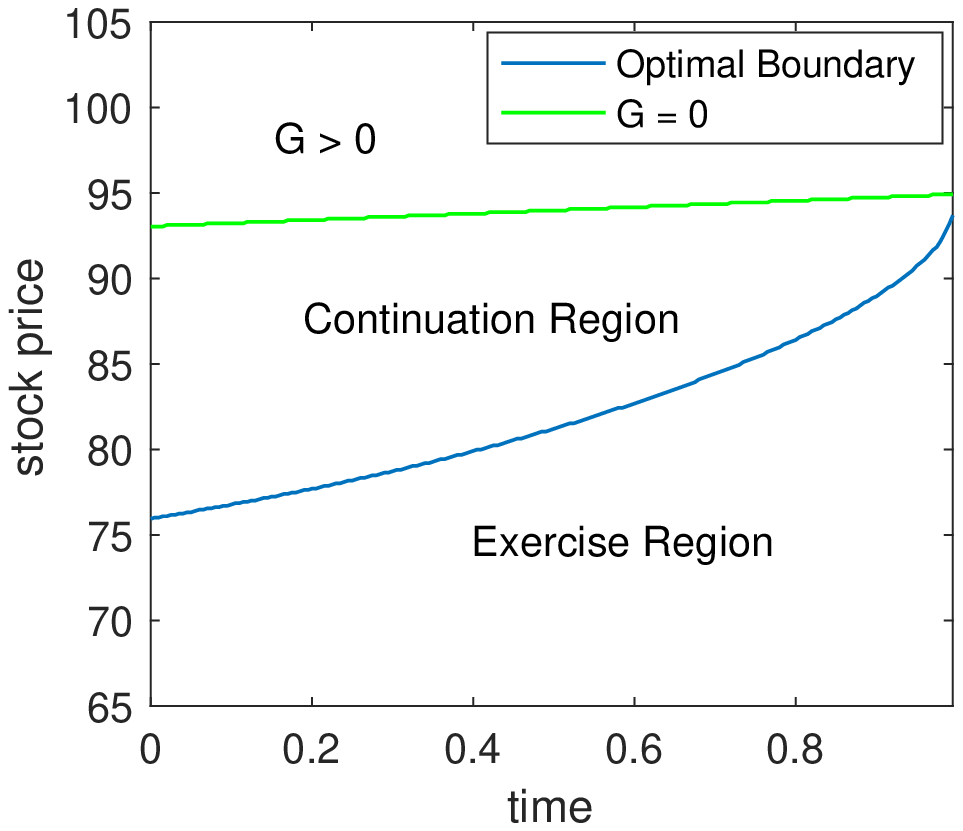}
\caption{}\label{boundary_put_sigma01_K100}
\end{subfigure} \hfill
\caption{Optimal boundaries for selling European call with $S_0 = 105$ (a) and put with $S_0 = 95$ (b) under normal distribution. The strike price $K = 100$ and maturity $T=1$, and the other parameters: $\mu = 0, \sigma _D = 0.1, r=0.1, \T =1, \sigma = 0.4$. The green line is defined by $G = 0$, and continuation region contains $G\geq 0$.}\label{boundaries_option}
\end{figure}

By direct calculation using the drive functions \eqref{drive_call} and \eqref{drive_put} for a call and a put, and comparing with the drive function  \eqref{Gstock}  for   the underlying asset, we arrive at the following parity result.

\begin{lemma}  (Call-Put  Parity) Consider a pair of  European call and put options with the same underlying $S$, strike $K$ and maturity $T$. Under model  \eqref{LogPrice} with the same distribution for $D$, the associated  drive functions satisfy
\begin{align}
G_{call}(t, x) - G_{put}(t, x)  = G_{stock} (t,x), \label{relation}
\end{align}for all $(t,x) \in [0,T]\times \mathbb{R}$.
\end{lemma}

 \begin{proposition} \label{parity}
Under the same distribution of $D$, consider  European-style call and put options with the same strike and maturity. 
 The optimal strategy to liquidate a long-call-short-put position is identical to the optimal strategy to sell the underlying stock under the same trading horizon.      \end{proposition}

\begin{proof}The delayed liquidation premium associated with the long-call-short-put position is given by
\begin{align}
L_{call-put}(t, x) &=\sup_{\tau \in\setT_{t,\bar{T}}}{\E} \left\{\int_t^\tau e^{-r(u-t)} \big(G_{call}(u, X_u)- G_{put}(u, X_u)\big)\,du\, |\, X_t = x \right\}\notag\\
&=\sup_{\tau \in\setT_{t,\bar{T}}}{\E} \left\{\int_t^\tau e^{-r(u-t)} G_{stock}(u, X_u)\,du \,|\, X_t = x \right\}
 \label{Optimal_Parity}
\end{align}
The last term is delayed liquidation premium associated with selling the stock $S$. 
\end{proof}

As an interesting consequence of  Proposition \ref{parity}, consider a pair of call and put with strike $K_1$ with maturity $T_1$, and  another pair with strike $K_2$ and maturity $T_2$, with $\bar{T}\le \min\{T_1,T_2\}$. The respective long-call-short-put positions will yield the same optimal timing strategy. This strategy is identical to that of selling the underlying stock $S$, which is independent of strike and maturity.

\section{Numerical Implementation}\label{sect-finitedifference}
We now summarize  the numerical scheme used to solve the  variational inequality satisfied by the value function  for the optimal trading problem \eqref{OptimalProb}. A finite difference scheme is used  to solve for the optimal liquidation boundaries. 
%

The numerical solution of the system of variational inequalities can be obtained by applying a finite-difference scheme with the use of the Projected-Successive-Over-Relaxation (PSOR) method. The solution of the resulting equations for value functions are solved by the Successive Over Relaxation (SOR) method. We refer to  \citet[Chap.9]{wilmottbook1995} for a detailed discussion on the projected SOR method.  In each SOR iterative step in finding the numerical approximation of the value functions, we simply take the maximum value between the approximated function value and asset payoff. Then the variational inequality \eqref{VIV} admits the    linear complimentarity form:
\begin{align}
\begin{cases} 
\begin{split}
\L V(t, x) \leq 0, \enspace V(t, x) & \geq f(t, x),  \quad (t, x) \in [0,\T) \times \R_+,  \\ 
\\ (\L V(t, x)) (f (t,x) - V (t, x)) &= 0,  \quad (t, x) \in [0,\T) \times \R_+,\\
\\ V (\T, x) &= f (\T, x),  \quad x \in \R_+.
\end{split}	
\end{cases} \label{VIg}				
\end{align}

We now consider the discretization of the partial differential equation $ \L V(t,x) =0$, over a bounded uniform grid with discretizations in  time ($\delta t =  {\T}/{N}$), and space    ($\delta x =  {(X_{\max} - X_{\min})}/{M})$, where $X_{\max}$ and $X_{\min}$ are the upper and lower bounds on the values of $x$ on the grid.  Applying  the Crank-Nicolson method for $x$-derivatives and backward difference for $t$-derivatives on the resulting equation leads to the grid equation:
\begin{align}
-\alpha_{i,j-1}V_{i-1,j-1} + (1 - \beta )V_{i,j-1} - \gamma_{i,j-1}V_{i+1,j-1} = \alpha_{i,j} V_{i-1,j} + (1 + \beta) V_{i,j} + \gamma_{i,j} V_{i+1,j}
\end{align}
where the coefficients are given by

\begin{align}
\begin{cases} 
\begin{split}
\alpha_{i,j} &= \delta t \left( \frac{\sigma^2}{4(\delta x)^2} - \frac{A(t_j,x_{i-1}) + A(t_j,x_{i+1})}{8\delta x} \right),  \\ 
\\ \beta &= - \frac{\delta t}{2} \left(r + \frac{\sigma^2}{(\delta x)^2} \right),\\
\\ \gamma_{i,j} &= \delta t \left( \frac{\sigma^2}{4(\delta x)^2} + \frac{A(t_j,x_{i-1}) + A(t_j,x_{i+1})}{8\delta x} \right),
\end{split}	
\end{cases} \label{coef_RS}		
\end{align}
for $i=1,2,...,M-1$ and $j=1,2,...,N-1$. The system to be solved  backward in time is 
\begin{align}
\mathbf{M_{j-1} ^ 1 V_{j-1}=r_j},\label{M1g}
\end{align}
where the right-hand side is 
\begin{align}
\mathbf{r_j = M_j ^2 V_j} +  \begin{bmatrix} \alpha_{1, j-1} V_{0, j-1}+\alpha_{1, j} V_{0, j} \\ 0 \\ \vdots \\0 \end{bmatrix} +  \begin{bmatrix}  0 \\ \vdots \\0 \\ \gamma_{M-1, j-1} V_{M, j-1}+\gamma_{M-1, j} V_{M, j} \end{bmatrix},
\end{align}
and
\begin{align}
\mathbf{M_{j} ^1} &= \left[ \begin{array}{cccccc}
1- \beta  & -\gamma_{1, j} & & & \\
-\alpha _{2, j} & 1- \beta & -\gamma_{2, j} & & \\
& -\alpha _{3, j} & 1- \beta & -\gamma_{3, j} & \\
& & \ddots & \ddots & \ddots \\
& & &- \alpha_{M-2, j} & 1- \beta & -\gamma_{M-2, j} \\
& & & & - \alpha_{M-1,j} & 1- \beta  \end{array} \right],\\
 \nonumber\\
\mathbf{M_{j} ^2} &= \left[ \begin{array}{cccccc}
1+ \beta & \gamma_{1,j} & & & \\
\alpha _{2, j} & 1+ \beta & \gamma_{2, j} & & \\
& \alpha _{3, j} & 1+ \beta  & \gamma_{3, j} & \\
& & \ddots & \ddots & \ddots \\
& & & \alpha_{M-2, j} & 1+ \beta  & \gamma_{M-2, j} \\
& & & & \alpha_{M-1, j} & 1+ \beta \end{array} \right],\\
 \nonumber\\
\mathbf{V_j} &=\begin{bmatrix} V_{1,j}, V_{2,j}, \hdots , V_{M-1,j} \end{bmatrix} ^T.
\end{align}
This  leads to a sequence of stationary complementarity problems. Since the trader can close her position   anytime prior to expiry, the value function $V(t, x)$ must satisfy the constraint
\begin{align}
V(t,x) \geq f (t,x),  \quad X_{\min} \leq x \leq X_{\max},\quad  0 \leq t \leq \T.
\end{align}
Correspondingly,  the discrete scheme can be written as 
\begin{align}
V_{i, j} \geq f_{i,j}, \quad 0 \leq i \leq M, \quad  0 \leq j \leq N.
\end{align}
Hence, at each time step $j \in \left\{1, 2, \hdots, N-1\right\}$,  we need to solve the linear complementarity problem
\begin{align}
\begin{cases} 
\begin{split}
\mathbf{M_{j-1} ^1 V_{j-1}} & \geq \mathbf{r_j}, \\
\\ \mathbf{V_{j-1}} & \geq \mathbf{f_{j-1}},   \\ 
\\ (\mathbf{M_{j-1} ^1 g_{j-1}} - \mathbf{r_j})^T (\mathbf{f_{j-1}} - \mathbf{g_{j-1}}) &= 0.  
\end{split}	
\end{cases} 			
\end{align}
In particular, our algorithm enforces the constraint of dominating the payoff function   as follows:
\begin{align}
V_{i,j-1}^{new}=\max \big\{V_{i, j-1} ^{old}, f_{i,j-1}\big\}.
\label{iterative}
\end{align}

The   projected SOR method is used to solve the linear system. Notice that the constraint is enforced at the same time as  $V_{i,j-1}^{(k+1)}$ is calculated in each iteration. Hence, at each time step $j$, the PSOR algorithm is to iterate (on $k$) the equations
\begin{align*} 
\begin{split}
V_{1, j-1}^{(k+1)} &= \max \big\{f_{1,j-1} \,,\, V_{1, j-1}^{(k)} + \frac{\omega}{1-\beta} [r_{1, j}-(1-\beta) V_{1, j-1}^{(k)}+\gamma_{1, j-1} V_{2, j-1}^{(k)}] \big\},\\
V_{2, j-1}^{(k+1)} &= \max \big\{f_{2,j-1} \,,\, V_{2,j-1}^{(k)} + \frac{\omega}{1-\beta} [r_{2, j}+\alpha_{2, j-1} V_{1, j-1}^{(k+1)}-(1-\beta) V_{2, j-1}^{(k)}+\gamma_{2, j-1} V_{3,j-1}^{(k)}] \big\},\\
\vdots\\
V_{M-1,j-1}^{(k+1)} &= \max \big\{f_{M-1,j-1} \,,\, V_{M-1,j-1}^{(k)} + \frac{\omega}{1-\beta} [r_{M-1, j}+\alpha_{M-1,j-1} V_{M-2,j-1}^{(k+1)}-(1-\beta) V_{M-1,j-1}^{(k)}] \big\},
\end{split}
\label{PSOR}					
\end{align*}
where $k$ is the iteration counter and $\omega$ is the overrelaxation parameter.
The iterative scheme starts from an initial point $\mathbf{V_{j-1} ^{\mathrm{(0)}}}$ and proceeds until a convergence criterion is met, such as
$|| \mathbf{V_{j-1} ^{\mathit{(k+\mathrm{1})}}} - \mathbf{V_{j-1} ^{\mathit{(k)}}} || < \epsilon ,$ where $\epsilon$ is a tolerance parameter.  The  optimal boundary is identified by  locating the grid points separating the exercise (or sell) region and continuation region, respectively defined by   $\{(t,x)\,:\, V(t,x)=f(t,x)\}$ and   $\{(t,x)\,:\, V(t,x) > f(t,x)\}$.

\section{Concluding Remarks}
We have studied the optimal timing to sell an asset or option when the underlying asset's log-price follows a randomized Brownian bridge. By incorporating the trader's prior belief on the terminal stock price and allowing it to be updated with new price information, this approach can shed light on the effect of belief on the optimal trading strategy.  We have explicitly derived  the price dynamics under different beliefs (e.g. two-point discrete, double exponential, and normal distributions), and analyzed the properties of the associated delayed liquidation premia via the drive functions. By numerically solving the underlying variational inequality, we obtain the optimal trading strategies expressed in terms of sell/hold boundaries or regions. In particular, we find that the optimal strategy for liquidating a stock can  admit connected or disconnected continuation/exercise regions under the two-point discrete distribution or double exponential distribution depending on the parameters values.

For future research,  a natural direction is to consider under  the current  stochastic framework  more sophisticated timing strategies for options or  other derivatives, such as futures (see \cite{LeungLiLiZheng2015}) and credit derivatives (see \cite{LeungLiu2012}). Other extensions include introducing additional stochastic factors that incorporate stochastic volatility and jumps to the Brownian bridge, and/or considering alternative distributions for the trader's prior belief.

\section{Proofs} \label{proofs}
In this section, we  provide the derivations for $a(t, x)$ given in Examples 2 (normal distribution) and 3 (double exponential distribution), respectively, in Sections  \ref{CompNormal} and \ref{CompDouble}. In addition, we provide the details for Proposition 5 and 6.

\subsection{Normal Distribution}\label{CompNormal}
If the trader's prior belief on the log-price follows a normal distribution, then  \eqref{a} can be expressed as
\begin{align*}
a(t, x) = \frac{I_1 (t, x)}{I_2 (t, x)},
\end{align*}
where
\begin{align*}
\begin{cases} 
\begin{split}
I_1 (t, x) &= \int_{-\infty}^\infty z\exp\left(z \frac{x- X_0}{\sigma^2(T-t)} -  z^2 \frac{t}{2T\sigma^2(T-t)} - \frac{(z-\mu)^2}{2\sigma_D^2} \right) dz,  \\ 
\\ 
I_2 (t, x) &= \int_{-\infty}^\infty \exp \left( z\frac{x - X_0}{\sigma^2(T-t)} -  z^2 \frac{t}{2T\sigma^2(T-t)} - \frac{(z-\mu)^2}{2\sigma_D^2} \right) dz.
\end{split}	
\end{cases} 	
\end{align*}
Next, we compute $I_1$ and $I_2$ explicitly. To facilitate the representation, we denote
\begin{align*}
\begin{cases} 
\begin{split}
\eta &\equiv \eta (t) = \frac{t}{2T\sigma^2 (T-t)} + \frac{1}{2\sigma_D ^2},  \nonumber\\ 
\\ 
b &\equiv b (t, x) = -\left( \frac{x - X_0}{\sigma^2 (T-t)}  + \frac{\mu}{\sigma_D ^2} \right),\\
\\ 
c &= \frac{\mu ^2}{2 \sigma_D ^2}.
\end{split}	
\end{cases}		
\end{align*}
Substituting the above notations into $I_1 (t,x)$ and $I_2 (t,x)$ and rearranging terms,  we have
\begin{align*}
I_2 (t, x) &= \int_{-\infty}^\infty \exp \left( -(\eta z^2 + b z + c) \right) dz.
\end{align*}
By applying a change of variables $y = \eta (z + \frac{b}{2 \eta })^2$, we obtain
\begin{align*}
I_1 (t, x) &= \int_{-\infty}^\infty z \exp \left( -( \eta  z^2 + b z + c) \right) dz\\
&= \int_{\infty}^0 \frac{1}{2 \eta }  \exp (\frac{b^2 - 4 \eta c}{4 \eta }) \exp(-y) dy + \int_{0}^\infty \frac{1}{2 \eta }  \exp (\frac{b^2 - 4 \eta c}{4 \eta }) \exp(-y) dy - \frac{b}{2\eta} I_2(t, x)\\
&= - \frac{b}{2\eta} I_2(t, x).
\end{align*}
It follows immediately that
\begin{align*}
a(t, x) &= \frac{I_1 (t, x)}{I_2 (t, x)} = - \frac{b}{2\eta} = T \frac{(x - X_0) \sigma_D ^2 + \mu \sigma^2 (T-t)}{t \sigma_D ^2 + T \sigma^2 (T-t)}.
\end{align*}
This, together with \eqref{AA}, gives us 
\begin{align*}
A(t,X_t) &= \frac{a(t, X_t) - (X_t - X_0)}{T - t}\\
&= \frac{\sigma_D ^2 (X_t - X_0) + T \sigma ^2 (\mu - (X_t - X_0))}{t \sigma_D ^2 + T \sigma ^2 (T-t)}.
\end{align*}

\subsection{Double Exponential Distribution}\label{CompDouble}
If the trader's prior belief on the log-price follows a double exponential distribution, then \eqref{a} can be expressed as
\begin{align*}
a(t, x) = \frac{N_1 (t, x) + N_2 (t, x)}{H_1 (t, x) + H_2 (t, x)},
\end{align*}
where
\begin{align*}
\begin{cases} 
\begin{split}
N_1 (t, x) &= \int_{-\infty}^\theta p_1 \lambda_1 z\exp\left(z \frac{x - X_0}{\sigma^2(T-t)} -  z^2 \frac{t}{2T\sigma^2(T-t)} - \lambda_1(\theta - z) \right) dz,  \\ 
\\ N_2 (t, x) &= \int_{\theta}^\infty p_2 \lambda_2 z\exp\left(z \frac{x - X_0}{\sigma^2(T-t)} -  z^2 \frac{t}{2T\sigma^2(T-t)} - \lambda_2(z - \theta ) \right) dz,  \\ 
\\ H_1 (t, x) &= \int_{-\infty}^\theta p_1 \lambda_1 \exp \left( z\frac{x - X_0}{\sigma^2(T-t)} -  z^2 \frac{t}{2T\sigma^2(T-t)} - \lambda_1(\theta - z)  \right) dz,\\
\\ H_2 (t, x) &= \int_{\theta}^\infty p_2 \lambda_2 \exp \left( z\frac{x - X_0}{\sigma^2(T-t)} -  z^2 \frac{t}{2T\sigma^2(T-t)} - \lambda_2(z - \theta ) \right) dz.
\end{split}	
\end{cases} 	
\end{align*}
Now we compute $N_1 (t,x), N_2 (t,x), H_1 (t,x)$ and $H_2 (t,x)$ explicitly. By completing the squares with a change of variables $y = \sqrt{2 \zeta}(z + \frac{b_1}{2 \zeta})$, we have
\begin{align*}
H_1 (t, x) &= \int_{-\infty}^\theta p_1 \lambda _1 \exp \left( -(\zeta z^2 + b_1 z + c_1) \right) dz\\
&=  p_1 \lambda _1 \frac{\sqrt{2 \pi}}{\sqrt{2 \zeta}} \exp (\frac{b_1^2 - 4 \zeta c_1}{4\zeta}) \frac{1}{\sqrt{2 \pi}} \int_{-\infty}^{\sqrt{2 \zeta}(\theta + \frac{b_1}{2 \zeta})} \exp \left( -\frac{y^2}{2} \right) dy \\
&=  p_1 \lambda _1 \frac{\sqrt{\pi}}{\sqrt{\zeta}} \exp (\frac{b_1^2 - 4\zeta c_1}{4\zeta}) \Phi(d_1),
\end{align*}
where 
\begin{align*}
\begin{cases} 
\begin{split}
\zeta &\equiv \zeta (t) = \frac{t}{2T\sigma^2 (T-t)},  \\ 
\\ 
b_1 &\equiv b_1 (t,x) = -\left( \frac{x - X_0}{\sigma^2 (T-t)}  + \lambda_1 \right),\\
\\ 
c_1 &= \lambda_1 \theta,\\
\\
d_1 &\equiv d_1 (t,x) = \sqrt{2\zeta (t)} \left(\theta + \frac{b_1 (t,x)}{2\zeta (t)} \right).
\end{split}	
\end{cases}		
\end{align*}
Similarly, by applying $y = \sqrt{2 \zeta}(z + \frac{b_2}{2 \zeta})$, we obtain
\begin{align*}
H_2 (t, x) &= \int_{\theta}^\infty p_2 \lambda _2 \exp \left( -(\zeta z^2 + b_2 z + c_2) \right) dz\\
&=  p_2 \lambda _2 \frac{\sqrt{2 \pi}}{\sqrt{2 \zeta}} \exp (\frac{b_2^2 - 4 \zeta c_2}{4\zeta}) \frac{1}{\sqrt{2 \pi}} \int_{\sqrt{2 \zeta}(\theta + \frac{b_2}{2 \zeta})}^\infty \exp \left( -\frac{y^2}{2} \right) dy\\
&= p_2 \lambda _2 \frac{\sqrt{\pi}}{\sqrt{\zeta}} \exp (\frac{b_2^2 - 4 \zeta c_2}{4\zeta}) \Phi(-d_2),
\end{align*}
where 
\begin{align*}
\begin{cases} 
\begin{split}
b_2 &\equiv b_2 (t,x) = -\left( \frac{x - X_0}{\sigma^2 (T-t)}  -\lambda_2 \right),\\
\\ 
c_2 &= -\lambda_2 \theta,\\
\\
d_2 &\equiv d_2(t,x) = \sqrt{2\zeta(t)}\left(\theta + \frac{b_2(t,x)}{2\zeta(t)}\right).
\end{split}	
\end{cases}		
\end{align*}

Next, we compute $N_1$ and $N_2$ in the numerator with changes of variables $y = \zeta(z + \frac{b_1}{2\zeta})^2$ and $y = \zeta(z + \frac{b_2}{2\zeta})^2$ respectively. 
\begin{align*}
N_1 (t, x) &= \int_{-\infty}^\theta p_1 \lambda _1 z \exp \left( -(\zeta z^2 + b_1 z + c_1) \right) dz\\
&= -\int_{\zeta(\theta + \frac{b_1}{2\zeta})^2}^{\infty} p_1 \lambda _1 \frac{1}{2\zeta} \exp (\frac{b_1^2 - 4\zeta c_1}{4\zeta}) \exp(-y) dy - \frac{b_1}{2\zeta} H_1(t, x)\\ 
&= - p_1 \lambda _1 \frac{1}{2\zeta} \exp \left( - \zeta(\theta + \frac{b_1}{2\zeta})^2 \right) \exp (\frac{b_1^2 - 4\zeta c_1}{4\zeta}) - p_1 \lambda _1 \frac{b_1}{2\zeta} \frac{\sqrt{\pi}}{\sqrt{\zeta}} \exp (\frac{b_1^2 - 4\zeta c_1}{4\zeta}) \Phi(d_1),
\end{align*}
and
\begin{align*}
N_2 (t, x) &= \int_\theta^{\infty}  p_2 \lambda _2 z \exp \left( -(\zeta z^2 + b_2 z + c_2) \right) dz\\
&= \int_{\zeta(\theta + \frac{b_2}{2\zeta})^2}^{\infty} p_2 \lambda _2 \frac{1}{2\zeta} \exp (\frac{b_2^2 - 4\zeta c_2}{4\zeta}) \exp(-y) dy -  \frac{b_2}{2\zeta} H_2 (t, x)\\
&= p_2 \lambda _2 \frac{1}{2\zeta} \exp \left( - \zeta(\theta + \frac{b_2}{2\zeta})^2 \right) \exp (\frac{b_2^2 - 4\zeta c_2}{4\zeta}) - p_2 \lambda _2 \frac{b_2}{2\zeta} \frac{\sqrt{\pi}}{\sqrt{\zeta}} \exp (\frac{b_2^2 - 4\zeta c_2}{4\zeta}) \Phi(-d_2).
\end{align*}

\subsection{Proof of Propositions 5 and 6}\label{Proof_proposition5}
Differentiating $G_{stock} (t, x)$ in \eqref{Gstock} yields that
\begin{align*}
\frac{\partial G_{stock}}{\partial x} (t, x) = \left( -r + \frac{1}{2} \sigma ^2 + \frac{(\sigma_D ^2 - T \sigma ^2)(x - X_0) + \mu \sigma ^2 T + \sigma _D ^2 - T \sigma ^2}{t \sigma _D ^2 + T \sigma ^2 (T - t)} \right) e^x,
\end{align*}
and
\begin{align*}
\frac{\partial ^2 G_{stock}}{\partial x^2} (t, x) = \left( -r + \frac{1}{2} \sigma ^2 + \frac{(\sigma_D ^2 - T \sigma ^2)(x - X_0) + \mu \sigma ^2 T + 2(\sigma _D ^2 - T \sigma ^2)}{t \sigma _D ^2 + T \sigma ^2 (T - t)} \right) e^x,
\end{align*}
and 
\begin{align*}
\frac{\partial G_{stock}}{\partial t} (t, x) = - \frac{\sigma_D ^2 - T \sigma ^2}{t \sigma_D ^2 + T \sigma^2 (T-t)} e^x A(t,x),
\end{align*}
and 
\begin{align*}
\frac{\partial ^2 G_{stock}}{\partial t^2} (t, x) =2 (\frac{\sigma_D ^2 - T \sigma ^2}{t \sigma_D ^2 + T \sigma^2 (T-t)})^2 e^x A(t,x),
\end{align*}
where $A(t,x)$ is given by \eqref{A_normal}. The results follow directly from the signs of these partial derivatives. The proof of Proposition \ref{normalt} is similar to Proposition \ref{normalx}.


 \begin{small}\bibliographystyle{apa}    
\bibliography{mybib_NEW}     
 
 \end{small}

\end{document}